\documentclass[a4paper,11pt]{article}

\usepackage{relsize}
\usepackage{nicefrac}
\usepackage{fullpage}
\usepackage{setspace}
\usepackage{hyperref}
\hypersetup{colorlinks=true,linkcolor=black,citecolor=blue}
\usepackage{amsmath,amsthm,amssymb}
\usepackage{mathtools}
\usepackage[titletoc,title]{appendix}
\usepackage{graphicx, color}
\usepackage{caption}
\usepackage{bbm}

\usepackage[normalem]{ulem}

\newtheorem{theorem}{Theorem}[section]

\newtheorem{corollary}[theorem]{Corollary}
\newtheorem{lemma}[theorem]{Lemma}
\newtheorem{open}{Open Question}

\newtheorem{proposition}[theorem]{Proposition}
\newtheorem{definition}[theorem]{Definition}
\newtheorem{remark}[theorem]{Remark}

\def\N{{\mathbb {N}}}

\def\R{{\mathbb {R}}}
\def\CC{{\mathsf {CC}}}

\def\eps{\varepsilon}
\def\norm{\R_{\geq 1}\cup\{\infty\}}

\renewcommand{\tilde}{\widetilde}
\newcommand{\Det}{\mathsf{CC}_{\mathtt{det}}}
\renewcommand{\CC}{\mathsf{CC}}
\newcommand{\Comp}{\mathsf{Comp}}
\newcommand{\Concat}{\mathsf{Concat}}
\newcommand{\Mean}{\mathsf{Mean}}
\newcommand{\Local}{\mathsf{Local}}
\newcommand{\Sp}{\mathsf{Sp}}
\newcommand{\EOL}{\mathsf{EoL}}
\newcommand{\conv}{\mathsf{conv}}
\newcommand{\para}{n,\varepsilon,\lambda_A,\lambda_B}
\newcommand{\param}{p,\para}
\newcommand{\parami}{\infty,\para}

\renewcommand{\L}{\mathsf{Loc}}
\newcommand{\PPAD}{\textsf{PPAD}}

\title{On Communication Complexity of Fixed Point Computation}
\author{Anat Ganor%
\thanks{Hebrew University of Jerusalem, Israel. This project has received funding from the European Research Council (ERC) under the European Union’s Horizon 2020 research and innovation programme (grant agreement No 740282).
	Email: anat.ganor@gmail.com}
\and Karthik C.\ S.\thanks{Weizmann Institute of Science, Israel.
   This work was supported by Irit Dinur's ERC-CoG grant 772839.
   Email: karthik0112358@gmail.com}
\and D{\"{o}}m{\"{o}}t{\"{o}}r P{\'{a}}lv{\"{o}}lgyi\thanks{MTA-ELTE Lend\"ulet Combinatorial Geometry Research Group, Institute of Mathematics, E\"otv\"os Lor\'and University (ELTE), Budapest, Hungary. Supported by the Lend\"ulet program of the Hungarian Academy of Sciences (MTA), under grant number LP2017-19/2017. Email: dom@cs.elte.hu}}

\date{}

\begin{document}

\maketitle

\begin{abstract}
  Brouwer's fixed point theorem states that any continuous function from a compact convex space to itself has a fixed point. Roughgarden and Weinstein (FOCS 2016) initiated the study of fixed point computation in the two-player communication model, 
where each player gets a function from $[0,1]^n$ to $[0,1]^n$,
and their goal is to find an approximate fixed point of the \emph{composition} of the two functions. They left it as an open question to show a lower bound of $2^{\Omega(n)}$ for the (randomized) communication complexity of this problem, in the range of parameters which make it a total search problem. 
We answer this question affirmatively.

Additionally, we introduce two natural fixed point problems in the two-player communication model.
\begin{itemize}
\item Each player is given a function from $[0,1]^n$ to $[0,1]^{n/2}$,
and their goal is to find an approximate fixed point of the \emph{concatenation}  of the functions.
\item Each player is given a function from $[0,1]^n$ to $[0,1]^{n}$,
and their goal is to find an approximate fixed point of the \emph{mean} of the functions.
\end{itemize}
We show a randomized communication complexity lower bound of $2^{\Omega(n)}$  for these problems (for some constant approximation factor). 

Finally, we initiate the study of finding a panchromatic simplex in a Sperner-coloring of a triangulation (guaranteed by Sperner's lemma) in the two-player communication model: A triangulation $T$ of the $d$-simplex is publicly known and one player is given a set $S_A\subset T$ and a coloring function from $S_A$ to $\{0,\ldots ,d/2\}$, and the other  player is given a set $S_B\subset T$ and a coloring function from $S_B$ to $\{d/2+1,\ldots ,d\}$, such that $S_A\dot\cup S_B=T$,
and their goal is to find a panchromatic simplex. We show a randomized communication complexity lower bound of $|T|^{\Omega(1)}$ for the aforementioned problem as well (when $d$ is large). On the positive side, we show that if $d\le 4$ then there is a deterministic protocol for the Sperner problem with $O((\log |T|)^2)$ bits of communication.
\end{abstract}

\clearpage

\section{Introduction}\label{sec:intro}

Fixed point theorems hold a very special place in Mathematics and is a cornerstone of Economic theory. In particular, Brouwer's fixed point theorem \cite{B12} is one of the most celebrated fixed point results 
in algebraic topology was famously used by Nash \cite{N51} to prove the existence of a mixed equilibrium in every finite game. 
Brouwer's fixed point theorem asserts that every continuous function from a compact convex space to itself has a fixed point. This result gives rise to a natural computational question -- given a continuous function find a fixed point 
(in a specified model of computation). 
This problem has been well-studied in various models of computation. 

Roughgarden and Weinstein \cite{RW16} initiated the study of 
distributed computation of approximate fixed points in the $\ell_\infty$ norm.
They studied the following task for two players: 
player $A$ gets a Lipschitz continuous function $f_A:[0,1]^n\rightarrow[0,1]^m$
and player $B$ gets a Lipschitz continuous function $f_B:[0,1]^m\rightarrow[0,1]^n$
where $m=O(n)$.
Their goal is to find an $\varepsilon$-approximate fixed point
of the composition of the two functions $f_{\Comp}:=f_B\circ f_A$, i.e., to find an $x\in[0,1]^n$ such that $\|f_B(f_A(x))-x\|_{\infty}\le \varepsilon$. In this paper we refer to the aforementioned problem\footnote{To be precise, Roughgarden and Weinstein \cite{RW16} studied the problem with an additional discretization parameter, and these details will be elaborated in Section~\ref{sec:introlb}.}, more generally for all $\ell_p$ norms, as the \emph{Composition Brouwer problem} in the $\ell_p$-norm and denote it by $\Comp_{\param}$, where $\lambda_A$ and $\lambda_B$ are the Lipschitz constants of $f_{A}$ and $f_B$ respectively.

In the communication model, there are multiple ways to capture a computational problem. In this regard, our first contribution is to introduce two other natural realizations of fixed point computation of a Brouwer function in the communication model, and show that they are all essentially equivalent.

One may see the Composition Brouwer problem arising naturally from the mathematical fact that the composition of two continuous functions is a continuous function. In the same spirit, we note that the mean of two continuous functions is a continuous function, and introduce the \emph{Mean Brouwer problem} in the $\ell_p$ norm (denoted by $\Mean_{\param}$), where player $A$ gets a  continuous function $f_A:[0,1]^n\rightarrow[0,1]^n$, player $B$ gets a  continuous function $f_B:[0,1]^n\rightarrow[0,1]^n$, and their goal is to find an $\varepsilon$-approximate fixed point 
of the mean of the two functions $f_{\Mean}:=\frac{f_A + f_B}{2}$.

Another natural way to partition the input function between the players in the communication model is to give each player part of the description of the input function. We introduce the \emph{Concatenation  Brouwer problem} in the $\ell_p$ norm (denoted by $\Concat_{\param}$), where player $A$ gets a  continuous function $f_A:[0,1]^n\rightarrow[0,1]^{n/2}$, player $B$ gets a  continuous function $f_B:[0,1]^n\rightarrow[0,1]^{n/2}$, and their goal is to find an $\varepsilon$-approximate fixed point 
of the concatenation of the two functions $f_{\Concat}:=(f_A,f_B)$.

We remark that all the aforementioned problems can be solved with $2^{O(n)}$ bits of communication (see Lemma~\ref{lem:upper}) if $\varepsilon>0$ and $\lambda_A,\lambda_B\ge 0$ are all constants (independent of $n$ and $p$). 
Our first result states that the above described Brouwer function problems are all equivalent up to polynomial factors. Throughout this paper we denote by $\CC$ the randomized communication complexity of a problem. 

\begin{theorem}\label{thm:three_problems}
Let $n\in\N$ be an even integer, $p\in\norm$, $\lambda_A,\lambda_B,\eps \ge 0$.
Then the following inequalities hold:
\begin{enumerate}
\item $ \CC(\Concat_{p,n,\eps,\lambda_A,\lambda_B}) \leq \CC(\Mean_{p,n,\eps/2,\lambda_A+1,\lambda_B+1}) $.
\item $ \CC(\Mean_{p,n,\eps,\lambda_A,\lambda_B}) \leq \CC(\Comp_{p,n,\eps,\frac{\lambda_A}{2}+1,\lambda_B+2}) $.
\item $ \CC(\Comp_{p,n,\eps,\lambda_A,\lambda_B}) \leq \CC(\Concat_{p,O(n),O_{\lambda_B}(\eps),4(\lambda_A+1),4(\lambda_B+1)}) $.
\end{enumerate}
\end{theorem}

Finally, notice that all three aforementioned problems are total  (i.e., an $\varepsilon$-approximate fixed point is guaranteed to exist), as continuity is preserved under composition, concatenation, and interpolation.

\subsection{Lower Bounds in the Total Regime}\label{sec:introlb}

In this subsection, we show that there is no small  communication  protocol   by showing a lower bound of $2^{\Omega(n)}$ bits for all the three problems, even when $\varepsilon>0$ and $\lambda_A,\lambda_B\ge 0$ are all constants. 

\begin{sloppypar}While Roughgarden and Weinstein \cite{RW16}  left it open to show lower bounds for $\Comp_{\parami}$, they were able to prove strong lower bounds for a variant where player $A$ gets a Lipschitz continuous function $f_A:[0,1]^n_{\alpha}\rightarrow[0,1]^m_{\alpha}$
and player $B$ gets a Lipschitz continuous  function $f_B:[0,1]^m_{\alpha}\rightarrow[0,1]^n_{\alpha}$,
where $\alpha$ is a discretization parameter, and their goal is to find an $\varepsilon$-approximate fixed point
of the composition of the two functions, \emph{if one exists}. They showed a lower bound of $2^{\Omega(n)}$ 
	on the deterministic\footnote{The deterministic lower bound in \cite{RW16} relies crucially in one of the steps on a lifting theorem of Raz and McKenzie \cite{RM99}. If we replace that lifting theorem with the one of G\"o\"os, Pittassi, and Watson \cite{GPW17} that was proven subsequent to \cite{RW16}, then we immediately extend the deterministic lower bound in \cite{RW16} to a randomized lower bound (by starting from the lower bound in \cite{B16} instead of \cite{HPV89}).} communication complexity of the above problem in the $\ell_\infty$ norm for a certain setting of parameters $\varepsilon,\alpha$, $\lambda_A$, and $\lambda_B$. Their proof strategy was to lift  the query complexity lower bounds for finding a fixed point of a Brouwer function into the communication model. However, for the setting of parameters for which their lower bound was shown, they could not  guarantee the existence of an $\varepsilon$-approximate fixed point\footnote{A different way to view this, is to say that their reduction from the Brouwer problem in the query model to $\Comp_{\parami}$ created many `artificial' $\varepsilon$-approximate fixed points, and thus finding an $\varepsilon$-approximate fixed point in the communication model did not help in finding an approximate fixed point in the query model.}. They left it as an open problem if one could extend their lower bound to a regime of parameters where one could guarantee an $\varepsilon$-approximate fixed point (hereafter referred to as the total regime).\end{sloppypar}

Babichenko and Rubinstein \cite{BR16} showed an exponential lower bound\footnote{One may wonder, if the lower bound for Nash equilibrium in \cite{BR16} (or even \cite{GR18}) would imply the lower bound for the fixed point problems considered in this paper by using the standard proof of Nash from Brouwer \cite{N51}. We argue in Appendix~\ref{sec:no}, that an immediate reduction of such a kind is unlikely to give strong lower bounds for the Euclidean norm.} in the total regime, for a version of the Brouwer problem in the communication model, building on the techniques of \cite{RW16}. 
	In this paper, we introduce the \emph{Local Brouwer problem} that captures the problem for which \cite{BR16} showed their lower bound. We reduce the Local Brouwer problem to the Composition Brouwer problem and thus resolve the open problem of \cite{RW16} (we reiterate that the open problem was to prove either deterministic or randomized lower bounds for the Composition Brouwer problem in the \emph{total} regime). 

\begin{theorem}\label{thm:main}
For $p\in\{2,\infty\}$ and some constants $\lambda_A,\lambda_B\ge 0$ and $\varepsilon>0$, we have 
$$\CC(\Comp_{\param})=2^{\Omega(n)}.$$
\end{theorem}

We emphasize that $\varepsilon,\lambda_A,$ and $\lambda_B$ in the above theorem are constants independent of $n$ and $p$. This implies that the previously mentioned naive protocol for $\Comp_{\param}$ matches the above lower bound  up to constant factors in the exponent.

The proof of the above theorem crucially uses the work of G\"o\"os and Rubinstein \cite{GR18}, who recently showed how to use the constant gadget size lifting theorem of G\"o\"os and Pitassi~\cite{GP14a} to obtain randomized communication lower bounds for the Local Brouwer problem. 

Also note that Theorem~\ref{thm:main} implies the lower bounds for the Composition Brouwer problem as defined in \cite{RW16} with the additional discretization parameter $\alpha$, even when $\alpha<\frac{2\varepsilon}{\lambda_A\lambda_B+1}$, which is the setting of parameters for the total regime (see Proposition~\ref{prop:comp} and Theorem~\ref{thm:total}). 

We remark that we can guarantee the existence of an approximate fixed point only in the Euclidean norm and the max norm due to known barriers on extension theorems for other norms \cite{N01}. We elaborate on this in Section~\ref{sec:total}.

Finally, the following is a simple corollary of Theorems~\ref{thm:three_problems} and \ref{thm:main} 
\begin{corollary}\label{cor:three_problems_lb}
For $p\in\{2,\infty\}$ and some constants $\lambda_A,\lambda_B\ge 0$ and $\varepsilon>0$, we have 
\begin{enumerate}
\item $\CC(\Mean_{\param})=2^{\Omega(n)}$.
\item $\CC(\Concat_{\param})=2^{\Omega(n)}$. 
\end{enumerate}
\end{corollary}

\subsection{Nash Equilibrium}

One of the main results of \cite{BR16} is that the randomized communication complexity of finding an {$\varepsilon\text{-Nash}$} equilibrium in two-player $N\times N$ games requires $N^{\Omega(1)}$ bits of communication\footnote{They also showed that the randomized communication complexity of finding an $\varepsilon$-Nash equilibrium in $N$-player binary action games requires $2^{\Omega(N)}$ bits of communication.}. Their result has received significant attention \cite{K17,R18,S18}, as it demonstrated a communication bottleneck for convergence to approximate Nash equilibrium via randomized uncoupled dynamics. The result of \cite{GR18} strengthens this result further and rules out $N^{2-o(1)}$ randomized communication protocols for finding an $\varepsilon$-Nash equilibrium in two-player $N\times N$ games. 

Utilizing Theorem~\ref{thm:main}, we provide below a modular (and relatively simpler) proof of the result of \cite{BR16} (see Appendix~\ref{sec:Nash} for a more detailed proof outline). Moreover, this affirms the original proof framework envisioned in \cite{RW16}. 

\begin{enumerate}
\item  We show an $\Omega(N)$ lower bound on the critical block sensitivity\footnote{See \cite{GR18} for definitions and a simple proof of the lower bound given in Step 1.} of the \emph{End of a Line} ($\EOL$) problem defined on the clique host graph on $N$ vertices. We replace the vertices in the clique with binary trees to obtain a lower bound (on critical block sensitivity) of $\Omega(\sqrt{N})$ for $\EOL$ on a host graph  on $N$ vertices of constant degree.
\item Next, we apply the simulation theorem of \cite{GP14a} on a constant sized gadget, to obtain a lower boundof $\Omega(\sqrt{N})$ for $\EOL$ in the communication model.
\item Then, we embed the input graph of $\EOL$ problem into a (continuous) Brouwer function in $O(\log N)$ dimensions in the Euclidean space using the embedding given in \cite{BR16} (which essentially follows from the one in \cite{R16}). This gives us a lower bound of $\Omega(\sqrt{N})$ on the randomized communication complexity of the Local Brouwer problem in $O(\log N)$ dimensions.
\item Now we apply the reduction in the proof of Theorem~\ref{thm:main}, to obtain a lower bound of $\Omega(\sqrt{N})$ on $\CC\left(\Comp_{2,O(\log N),\varepsilon,O(\lambda),O(\lambda)}\right)$, for some constants $\varepsilon$ and $\lambda$.
\item Finally, we use the imitation gadget\footnote{Given inputs $f_A$ and $f_B$ to players $A$ and $B$ respectively, they build utility functions $u_A$ and $u_B$ over the action space $[0,1]^n$ and $[0,1]^m$ respectively as follows: $u_A(x,y)=-\|f_A(x)-y\|_2^2$ and $u_B(x,y)=-\|x-f_B(y)\|_2^2$.} given in \cite{RW16} to reduce\footnote{We need to discretize the space $[0,1]^n$ using the discretization parameter $\alpha$, where $\alpha$ is smaller than $c\varepsilon/\lambda^2$, for some large constant $c$.} $\Comp_{2,O(\log N),\varepsilon,O(\lambda),O(\lambda)}$ to that of finding an $\varepsilon^{O(1)}$-approximate Nash equilibrium in two-player $N'\times N'$ game, where $N'=N^{O(1)}$. This gives us the lower bound of \cite{BR16}.
\end{enumerate}

First, we remark that the above proof strategy can only give us $N^{\Omega(1)}$ lower bounds and thus cannot be used to obtain the lower bound given by \cite{GR18}; for instance, we lose a polynomial factor in Step 4 (i.e., Theorem~\ref{thm:main}). Second, we note that none of the $\text{non-trivial}$ techniques developed in \cite{GR18} (i.e., proving $\tilde \Omega(N)$ lower bound on the critical block sensitivity of $\EOL$ on host graphs on $N$ vertices of constant degree, and the `doubly-local' embedding of $\EOL$ into a Brouwer function) are used in the above proof. We merely use the very nice idea of applying the simulation theorem of \cite{GP14a} to obtain randomized communication complexity lower bounds for $\EOL$ problem. Third, we remark that in the proofs of both \cite{BR16} and \cite{GR18}, steps 3-5 in the above proof strategy are delicately intertwined and thus the above proof is arguably easier to follow. Finally, we note that from the lower bound on Composition Brouwer in Step 4, we can also obtain the same lower bound as \cite{BR16} for the randomized communication complexity of finding an  $\varepsilon$-Nash equilibrium in $N$-player binary action games as well (see \cite{RW16} for details). 

It remains an interesting open question to find a more straightforward proof for the lower bound on the communication complexity of finding an $\varepsilon$-Nash equilibrium (ideally with no simulation theorems involved). 
 A small step in this direction was shown by \cite{GK18}. 
We discuss  some possibilities via connections to Hex games in Appendix~\ref{sec:hex}.

\subsection{Sperner Problem}

Sperner’s lemma \cite{S28} is used to show existence of solutions in many game-theoretic problems such as envy-free cake cutting \cite{II99,S99}, independent transversal problem (of forming committees for example) \cite{H11}, hyper graph extension of Hall’s theorem \cite{AH00}. Thus, modeling these results in the communication model sheds insight into the amount of interaction needed between the various agents involved in order to reach an agreement. 

We initiate the study of the computational problem associated with Sperner's lemma in the communication model. 
Let $T$ be a triangulation of the unit $d$-simplex $\Delta:=\conv(v_0,\ldots,v_d)$ (i.e., a subdivision of $\Delta$ into subsimplices; $T$ here would be the union of the vertex set of these subsimplices). A  coloring $c:T\to\{0,\ldots ,d\}$ is said to be a Sperner-coloring if $c(v_i)=i$ for all $i\in \{0,\ldots ,d\}$ and every $x\in T$ gets the color of one of the vertices of the smallest face of $\Delta$ that contains $x$. Sperner's lemma asserts that in every Sperner-coloring of a triangulation of $\Delta$, there exists a panchromatic $d$-simplex. The natural computational problem that is associated with Sperner's lemma is as follows: Given a coloring of a fixed triangulation of $\Delta$, find a panchromatic $d$-simplex (or a point in $T$ that violates Sperner-coloring). This problem has previously been studied in the query model \cite{CS98,D06,FISV09} and the Turing machine model \cite{P94,G01,CD09}.

We introduce the \emph{Concatenation Sperner problem} (denoted by $\Sp_{d,n}^t$) in the two-player communication model, where a triangulation $T$ (of $n$ points) of the unit $d$-simplex is publicly known, player $A$ is given a set $S_A\subset T$ and a coloring function $c_A:S_A\to \{0,\ldots ,t-1\}$, and player $B$ is given a set $S_B\subset T$ and a coloring function $c_B:S_B\to \{t,\ldots ,d\}$. Their goal is to find a panchromatic $d$-simplex in the triangulation or a point $x\in T$ that violates the assumption that $S_A\dot\cup S_B=T$. Note that with two bits of communication the players  can verify if the coloring of $T$ given together by $c_A$ and $c_B$ is a Sperner-coloring.

Our first result on this problem is on the positive side:
\begin{theorem}\label{thm:d4}
For every $d\in\mathbb{N}$, there is a deterministic protocol for $\Sp_{d,n}^{d-1}$ with $O(\log^2 n)$ bits of communication.
\end{theorem}

An immediate corollary of the above theorem is that if  $d\le 4$, then for all $t\in\{0,\ldots ,4\}$, there is a deterministic protocol for $\Sp_{d,n}^t$ with $O(\log^2 n)$ bits of communication (see Corollary~\ref{cor:d4}).

Additionally, the proof of the above theorem can be modified to give an $O(\log^2 n)$ communication deterministic protocol for the following three-player problem: A triangulation (of size $n$) of the unit $2\text{-simplex}$ (a planar triangle) is publicly known, each player is given a subset of the triangulation points corresponding to one of the three color classes, and their goal is to find a panchromatic triangle (see Corollary~\ref{cor:3-Sp} for a formal statement). Such an efficient protocol is in stark contrast to the query model and the Turing machine model where the equivalent Sperner problem is known to be hard (see \cite{CS98} and \cite{CD09} respectively). We highlight that the protocol critically uses the perks of the communication model, that each player has unlimited computation power (which is not allowed in the Turing machine model), and that each player knows part of the total input (which does not hold in the query model).

However, the Concatenation Sperner problem admits no efficient protocol for large $d$ as we show below. 
\begin{theorem}\label{thm:SpLB}
For large enough $d$ (i.e., $d:=\Omega(\log n)$), we have $\Sp_{d,n}^{d/2}=n^{\Omega(1)}$.
\end{theorem}

The proof of the above theorem follows by a reduction from the Composition Brouwer problem to the Concatenation Sperner problem, and then applying the lower bound from Theorem~\ref{thm:main}. 

\subsection{Related Works}
We already discussed the known results on the fixed point problem in the communication complexity model.
Next we briefly mention the literature on the fixed point problem in other models of computation.

\noindent\textbf{Query Complexity.} 
In the query model, the task is to find a fixed-point of a continuous function $f:[0,1]^n\to[0,1]^n$, 
where a query algorithm can only obtain information about $f$ by queries to the value of $f$ at points in $[0,1]^n$. 
The general research issue is to identify bounds on the number of queries needed to find a fixed-point, 
subject to the assumption that $f$ belongs to some given class of functions (for instance, piecewise linear functions).
The query complexity of computing a constant approximate fixed point in the \emph{max norm} 
was studied by Hirsch et al.~\cite{HPV89} in the deterministic setting.
Recently, Babichenko \cite{B16} extended their lower bounds to the randomized setting. 
Rubinstein \cite{R16} extended this to the case of constant approximate fixed point computation in the \emph{Euclidean norm}. 
Finally, note that tight randomized query lower bounds have been obtained by Chen and Teng \cite{CT07} for the fixed point computation of Brouwer's functions in fixed dimension. 

\noindent\textbf{Computational Complexity.} 
In this model of computation, an arithmetic circuit representing the function 
(can be seen as succinct encoding of the truth table) is provided to a Turing machine as input and the complexity measure is the number of steps the machine should run in order to find the fixed point of the function.
The computational complexity of computing an approximate fixed point in the \emph{max norm} was shown to be \PPAD-complete for exponentially small approximation parameters by Papadimitriou~\cite{P94}. 
A decade later, Chen et al.~\cite{CDT09} showed that computing an approximate fixed point in the {max norm} was \PPAD-complete for polynomial approximation parameter. 
This was recently improved to constant approximation by Rubinstein \cite{R15}. 
Finally, Rubinstein~\cite{R16} showed that computing a constant approximate fixed point in the \emph{Euclidean norm} is \PPAD-complete.
The computational complexity of computing a \emph{near} fixed point in the \emph{max norm} was shown to be \textsf{FIXP}-complete by Etessami and Yannakakis \cite{EY10}.

\subsection{Organization of the Paper}
In Section~\ref{sec:prelim} we define some notions and introduce notations that will be used throughout the paper. In Section~\ref{sec:problems},  we formally introduce the Brouwer problems that we study in this paper and  in Section~\ref{sec:equivalence} we prove Theorem~\ref{thm:three_problems}. In Section~\ref{sec:total} we compute the setting of parameters wherein the Brouwer problems are total and in Section~\ref{sec:lowerbound} we show Theorem~\ref{thm:main}.
Finally, in Section~\ref{sec:sperner} we introduce the Sperner problem that we study in this paper and prove Theorems~\ref{thm:d4}~and~\ref{thm:SpLB}.

\section{Preliminaries}\label{sec:prelim}
In this section we give some basic definitions, propositions and notations used throughout the paper.

\begin{definition}[Normalized p-norm]
For $p\in\R_{\geq 1}$, the normalized $p$-norm $\ell_p$ of $x\in\R^n$ is
\[ \|x\|_p = \left( \frac{1}{n}\cdot \sum_{i\in[n]}\left|x_i\right|^p \right)^{\nicefrac{1}{p}} .\]
\end{definition}

\begin{definition}[The max norm]
The max norm $\ell_\infty$ of $x\in\R^n$ is
\[ \|x\|_\infty = \max_{i\in[n]} \{\left|x_i\right|\} .\]
\end{definition}

\noindent Note that for every $p < p'$, $ \|x\|_p \leq \|x\|_{p'} \leq \|x\|_\infty $.
Throughout the paper, whenever we use the notation $\|\cdot\|_p$ without specifying $p$ explicitly, 
$p$ should be clear from the context.

The following proposition will be used later.
\begin{proposition}\label{prop:inequality0}
Let $p\in\norm$, $n,r\in\mathbb N$. Let $x:=(x_1,\ldots ,x_r),y:=(y_1,\ldots ,y_r)\in [0,1]^{nr}$ where for all $i\in[r]$ we have $x_i,y_i\in [0,1]^n$. We have $$\sum_{i\in [r]}\|x_i-y_i\|_p\ge   \|x-y\|_p.$$
\end{proposition}
\begin{proof}The statement is obvious for $p=\infty$. So we focus on finite $p\ge 1$.
\begin{align*}
 \|x-y\|_p&=\left(\frac{1}{r}\cdot \sum_{i\in[r]}\|x_i-y_i\|_p^p\right)^{1/p}\\
 &\le \left(\max_{i\in[r]}\|x_i-y_i\|_p^p\right)^{1/p}\\
 &= \max_{i\in[r]}\|x_i-y_i\|_p\\
  &\le \sum_{i\in[r]}\|x_i-y_i\|_p\qedhere
 \end{align*}
\end{proof}

\begin{proposition}\label{prop:inequality}
Let $p\ge 1$,  $n,r\in\mathbb N$. Let $x:=(x_1,\ldots ,x_r),y:=(y_1,\ldots ,y_r)\in [0,1]^{nr}$ where for all $i\in[r]$ we have $x_i,y_i\in [0,1]^n$. For any $i\in [r]$ we have $\|x_i-y_i\|_p\le r^{1/p}\cdot \|x-y\|_p$.
\end{proposition}
  \begin{proof}  Fix $i\in[r]$. We have:
\begin{align*}
  \|x_i-y_i\|_p^p  &= r\cdot \|(0^{(i-1)n},x_i,0^{(r-i)n})-(0^{(i-1)n},y_i,0^{(r-i)n})\|_p^p\\
  &\le r\cdot \|(x_1,\ldots ,x_r)-(y_1,\ldots ,y_r)\|_p^p\\
  &=r\cdot \|x-y\|_p^p. \qedhere
 \end{align*} 
  \end{proof}
  
In the above proposition, if $p=\infty$ then we have $\|x_i-y_i\|_\infty\le  \|x-y\|_\infty$.

\begin{definition}[Lipschitz constant]
Let $p\in\norm$, $\lambda\geq0$ and let $A\subseteq\R^n$ be a non-empty set.
A function $f:A\rightarrow \R^m$ is $\lambda$-Lipschitz in $\ell_p$-norm space if for all $x,y\in A$,
\[ \|f(x)-f(y)\|_p \leq \lambda \|x-y\|_p. \]
\end{definition}

\noindent If $p$ is clear form the context we say, for simplicity, that the function is $\lambda$-Lipschitz.

\section{Brouwer Fixed Point Communication Problems}\label{sec:problems}

In this section, we study how fixed point computation can be realized in the communication model. To this effect we revisit the problem of finding a fixed point in the composition of two Brouwer functions introduced by Roughgarden and Weinstein \cite{RW16}, and additionally introduce two new fixed point communication  problems.

\subsection{Fixed Points of Composition of Brouwer Functions}
\emph{The composition of two continuous functions is a continuous function.} Based on this fundamental mathematical statement,  Roughgarden and Weinstein \cite{RW16} introduced
the following definition of the distributed version of finding an approximate fixed point of composed functions  for the two-player case\footnote{The problem was introduced for the $\ell_\infty$-metric in \cite{RW16}, but we address the problem in this paper for all $\ell_p$-metrics.}. We denote the randomized communication complexity of this problem by
$\CC(\Comp_{\param})$. 

\begin{definition}[Composition Brouwer Problem \cite{RW16}]\label{def:AFP-RW}
Let $p\in\norm$, $n,m\in\N$, where $m=O(n)$,  and $\lambda_A,\lambda_B,\eps \ge 0$.
The   \emph{Composition Brouwer Problem} 
for two players $A$ and $B$ is as follows.
Let $p,n,m,\lambda_A,\lambda_B$, and $\eps$ be publicly known parameters. 
Player $A$ gets  a $\lambda_A$-Lipschitz function\footnote{Note that the input to each player is of infinite size/description. However, this is not an issue as the communication complexity of this problem is bounded above.} $f_A:[0,1]^n\rightarrow [0,1]^m$. 
Player $B$ gets  a \mbox{$\lambda_B$-Lipschitz} function $f_B:[0,1]^m\rightarrow [0,1]^n$. Let $f_{\Comp}:[0,1]^n\to[0,1]^n$ be defined as follows: for all $x\in[0,1]^n$, $f_{\Comp}(x)=f_B(f_A(x))$.
Their goal is to output any $x\in [0,1]^n$ such that 
\[ \|f_{\Comp}(x) - x\|_p \leq \eps .\]
\end{definition}

We would like to remark here that \cite{RW16} additionally parameterize the above problem using a discretization parameter $\alpha$, and ask to output an $\varepsilon$-approximate fixed point $x$ on the $\alpha$ discretized hypercube. However, our formulation is arguably cleaner, and we use it throughout the paper.

\begin{proposition}[Roughgarden and Weinstein\footnote{They state the proposition for $\ell_\infty$ norm, but the same proof works for all $\ell_p$ norms.} \cite{RW16}]\label{prop:comp}
Let $f_{A},f_B,$ and $f_{\Comp}$ be as in Definition~\ref{def:AFP-RW}. 
Let  $\lambda_A,\lambda_B,$ and $\lambda_{\Comp}$ be their respective Lipschitz constants. Then we have  $\lambda_{\Comp} \le \lambda_A\cdot \lambda_B$.
\end{proposition}


\subsection{Fixed Points of Concatenation of Functions}
\emph{The concatenation of two continuous functions is a continuous function.} Based on this basic mathematical statement, we introduce a new fixed point problem that comes up naturally in the context of communication complexity.
We call this problem the Concatenation Brouwer Problem
and denote its randomized communication complexity by
$\CC(\Concat_{\param})$. 

\begin{definition}[Concatenation Brouwer Problem]\label{def:concat}
Let $p\in\norm$, $n\in\N$ be an even number, 
 $\lambda_A,\lambda_B,\eps \ge 0$.
The   \emph{Concatenation Brouwer Problem} 
for two players $A$ and $B$ is as follows.
Let $p,n,m,\lambda_A,\lambda_B$, and $\eps$ be publicly known parameters. 
Player $A$ gets  a $\lambda_A$-Lipschitz function $f_A:[0,1]^n\rightarrow [0,1]^{\nicefrac{n}{2}}$. 
Player $B$ gets a  $\lambda_B$-Lipschitz function $f_B:[0,1]^n\rightarrow [0,1]^{\nicefrac{n}{2}}$. Let $f_{\Concat}:[0,1]^n\to[0,1]^n$ be defined as follows: for all $x\in[0,1]^n$, $f_{\Concat}(x)=(f_A(x),f_B(x))$.
Their goal is to output any $x\in [0,1]^n$ such that 
\[ \|f_{\Concat}(x) - x\|_p \leq \eps .\]
\end{definition}

We have a proposition below for Concatenation of functions, similar to Proposition~\ref{prop:comp}.
\begin{proposition}\label{prop:concat}
Let $f_{A},f_B,$ and $f_{\Concat}$ be as in Definition~\ref{def:concat}. 
Let  $\lambda_A,\lambda_B,$ and $\lambda_{\Concat}$ be their respective Lipschitz constants. Then we have  $\lambda_{\Concat} \le \|\lambda_{A,B}\|_p$, where $\lambda_{A,B}=(\lambda_A,\lambda_B)$.
\end{proposition}
\begin{proof}
Fix distinct $x,y\in[0,1]^n$ such that $\|f_{\Concat}(x)-f_{\Concat}(y)\|_p=\lambda_{\Concat}\cdot \|x-y\|_p$. We have:
\begin{align*}
\lambda_{\Concat}\cdot \|x-y\|_p=\|f_{\Concat}(x)-f_{\Concat}(y)\|_p&=\left(\|f_{A}(x)-f_{A}(y)\|_p^p+\|f_{B}(x)-f_{B}(y)\|_p^p\right)^{1/p}\\
&\le \left((\lambda_A^p+\lambda_B^p)\cdot \|x-y\|_p^p\right)^{1/p}\\
&= \|\lambda_{A,B}\|_p\cdot \|x-y\|_p\qedhere
\end{align*}
\end{proof}

\subsection{Fixed Points of Mean of Brouwer Functions}
\emph{The mean of two continuous functions is a continuous function.} Based on this fundamental mathematical statement about functions over vector spaces\footnote{To be precise, the statement is true is for functions over any vector space where scaling and addition are continuous on the corresponding topology.}, we introduce the following fixed point problem which captures geometric smoothening of the mean operator. 
We call this problem the Mean Brouwer problem
and denote its randomized communication complexity 
$\CC(\Mean_{\param})$.

\begin{definition}[Mean Brouwer Problem]\label{def:interpol}
Let $p\in\norm$, $n,m\in\N$, where $m=O(n)$,  and $\lambda_A,\lambda_B,\eps \ge 0$.
The   \emph{Mean Brouwer Problem} 
for two players $A$ and $B$ is as follows.
Let $p,n,m,\lambda_A,\lambda_B$, and $\eps$ be publicly known parameters. 
Player $A$ gets a $\lambda_A$-Lipschitz function $f_A:[0,1]^n\rightarrow [0,1]^{n}$. 
Player $B$ gets a $\lambda_B$-Lipschitz  function $f_B:[0,1]^n\rightarrow [0,1]^{n}$. Let $f_{\Mean}:[0,1]^n\to[0,1]^n$ be defined as follows: for all $x\in[0,1]^n$ and  $i\in[n]$, $f_{\Mean}(x)_i=\frac{f_A(x)_i+f_B(x)_i}{2}$.
Their goal is to output any $x\in [0,1]^n$ such that 
\[ \|f_{\Mean}(x) - x\|_p \leq \eps .\]
\end{definition}

We remark here that in the above definition we could define $f_{\Mean}$ in a more general way: for every integers $p\ge 0,q>0$ such that $p\le q$, let  $f_{\Mean}^{p/q}:[0,1]^n\to[0,1]^n$ be defined as follows: for all $x\in[0,1]^n$ and  $i\in[n]$, $f_{\Mean}^{p/q}(x)_i=\frac{p}{q}\cdot f_A(x)_i + \left(1-\frac{p}{q}\right)\cdot f_B(x)_i$. The results in this paper could be extended to this more general definition, but we skip doing so, for the sake of brevity.

Finally, we have a proposition below for Mean of functions, similar to Propositions~\ref{prop:comp}~and~\ref{prop:concat}.
\begin{proposition}\label{prop:inter}
Let $f_{A},f_B,$ and $f_{\Mean}$ be as in Definition~\ref{def:interpol}. 
Let  $\lambda_A,\lambda_B,$ and $\lambda_{\Mean}$ be their respective Lipschitz constants. Then we have  $\lambda_{\Mean} \le \frac{\lambda_A+\lambda_B}{2}$.
\end{proposition}
\begin{proof}
Fix distinct $x,y\in[0,1]^n$ such that $\|f_{\Mean}(x)-f_{\Mean}(y)\|_p=\lambda_{\Mean}\cdot \|x-y\|_p$. We have:
\begin{align*}
\lambda_{\Mean}\cdot \|x-y\|_p= \|f_{\Mean}(x)-f_{\Mean}(y)\|_p&\le\frac{1}{2}\cdot\left(\|f_{A}(x)-f_{A}(y)\|_p+\|f_{B}(x)-f_{B}(y)\|_p\right)\\
&\le \frac{\lambda_A+\lambda_B}{2}\cdot \|x-y\|_p\qedhere
\end{align*}
\end{proof}

\section{Equivalence of Composition, Concatenation, and Mean Brouwer Problems}\label{sec:equivalence}

In this section, we prove the equivalence between the three Brouwer problems (up to polynomial factors) that we  introduced in Section~\ref{sec:problems}.

First we prove an upper bound on the (deterministic) communication complexity of these problems. 

\begin{lemma}\label{lem:upper}
Let $n\in\N$, $p\in\norm$, and $\lambda_A,\lambda_B,\eps> 0$ all  be fixed constants.
It holds that 
\[ \CC\left(\Comp_{p,n,\eps,\lambda_A,\lambda_B}\right)=2^{O(n)}. \]
Moreover, we have that $\CC\left(\Concat_{p,n,\eps,\lambda_A,\lambda_B}\right)=2^{O(n)}$ and  $\CC\left(\Mean_{p,n,\eps,\lambda_A,\lambda_B}\right)=2^{O(n)}$ as well.
\end{lemma}
\begin{proof}
Since $\lambda_A$ and $\lambda_B$ are both bounded above by constants, then the Lipschitz constants of $f_\Comp,f_\Mean,$ and $f_\Concat$ are all bounded above by some constant (see Propositions~\ref{prop:comp},~\ref{prop:concat},~\ref{prop:inter}). 

Let $\delta:=\frac{\varepsilon}{1+\lambda_A\cdot \lambda_B}$. Following a simple packing argument (for example, see Lemma 16 in \cite{DKL19}), we have that there is a \emph{fixed} discrete set $T^*\subseteq [0,1]^n$ of size $(1+\nicefrac{2}{\delta})^n$ such that for every $x\in[0,1]^n$, there exists $y\in T^*$ for which $\|x-y\|_p\le \delta$. The protocol then is to evaluate $f_\Comp$ (or $f_\Mean$ or $f_\Concat$ respectively) on all the points in $T^*$ (where each coordinate of every point is specified up to some constant digits of precision (depending on $\varepsilon,\lambda_A,$ and $\lambda_B$). The total communication in this protocol is $O(n)\cdot (1+\nicefrac{2}{\delta})^n=2^{O(n)}$ bits. We claim below that there exists an $\varepsilon$-approximate fixed point of $f_{\Comp}$ in $T^*$.

 Let $x_{{\Comp}}^*$ be a fixed point of $f_\Comp$ (i.e., $f_\Comp(x_{\Comp}^*)=x_{\Comp}^*$). By the construction of $T^*$, there exists $y\in T^*$ such that $\|x_{\Comp}^*-y\|_p\le \delta$. We show below that $y$ is an $\varepsilon$-approximate fixed point of $f_{\Comp}$. 
\begin{align*}
\|f_{\Comp}(y)-y\|_p&\le \|f_{\Comp}(y)-x_{\Comp}^*\|_p+\|y-x_{\Comp}^*\|_p\\
&=\|f_{\Comp}(y)-f_{\Comp}(x_{\Comp}^*)\|_p+\|y-x_{\Comp}^*\|_p\\
&\le (1+\lambda_A\cdot \lambda_B)\cdot \|y-x_{\Comp}^*\|_p\\
&\le (1+\lambda_A\cdot \lambda_B)\cdot \|y-x_{\Comp}^*\|_p\\ &\le (1+\lambda_A\cdot \lambda_B)\cdot \delta=\varepsilon
\end{align*}
A similar argument works for showing $\CC\left(\Concat_{p,n,\eps,\lambda_A,\lambda_B}\right)=2^{O(n)}$ and  $\CC\left(\Mean_{p,n,\eps,\lambda_A,\lambda_B}\right)=2^{O(n)}$ as well. 
\end{proof}
 
For the rest of this section, we omit $p$  from the notations, as all the results hold for any fixed value $p\in\norm$. 
The proof of Theorem~\ref{thm:three_problems}  follows from the next three lemmas.

\begin{lemma}
Let $n\in\N$ be an even integer and $\lambda_A,\lambda_B,\eps\ge 0$.
It holds that 
\[ \CC\left(\Concat_{n,\eps,\lambda_A,\lambda_B}\right) \leq \CC\left(\Mean_{n,\eps/2,(\lambda_A+1),(\lambda_B+1)}\right). \]
\end{lemma}

\begin{proof}
Player $A$ gets $f_A:[0,1]^n\rightarrow [0,1]^{\nicefrac{n}{2}}$.
Player $B$ gets $f_B:[0,1]^n\rightarrow [0,1]^{\nicefrac{n}{2}}$.
Define $g_A,g_B:[0,1]^n\rightarrow [0,1]^{n}$ for every $x=(x_1,x_2)\in [0,1]^n$ as
\[ g_A(x) = \left( f_A(x),x_2 \right) \text{    and    }  g_B(x) = \left( x_1,f_B(x) \right) .\]

We now show that if the Lipschitz constant of $f_A$ is $\lambda_A$ then $g_A$ is at most $(\lambda_A+1)$-Lipschitz.\allowdisplaybreaks
\begin{align*}
\|g_A(x)-g_A(y)\|&=\frac{1}{2}\cdot \|(f_A(x)-f_A(y),x_2-y_2)\|\\
&\le \frac{1}{2}\cdot \left(\|(f_A(x)-f_A(y),0^{\nicefrac{n}{2}})\|+\|(0^{\nicefrac{n}{2}},x_2-y_2)\|\right)\\
&\le \|f_A(x)-f_A(y)\|+\|x-y\|\\
&\le (\lambda_A+1)\cdot \|x-y\|,
\end{align*}
where we used the triangle inequality in the first inequality above, and we used Propositions~\ref{prop:inequality0}~and~\ref{prop:inequality} in the second inequality. 

\begin{sloppypar}Similarly, we show that if the Lipschitz constant of $f_B$ is $\lambda_B$ then $g_B$ is at most \mbox{$(\lambda_B+1)$-Lipschitz}.
\begin{align*}
\|g_B(x)-g_B(y)\|&=\frac{1}{2}\cdot \|(x_1-y_1,f_B(x)-f_B(y))\|\\
&\le \frac{1}{2}\cdot \left(\|(x_1-y_1,0^{\nicefrac{n}{2}})\|+\|(0^{\nicefrac{n}{2}},f_B(x)-f_B(y))\|\right)\\
&\le \|f_B(x)-f_B(y)\|+\|x-y\|\\
&\le (\lambda_B+1)\cdot \|x-y\|,
\end{align*}\end{sloppypar}
where we used the triangle inequality in the first inequality above, and we used Propositions~\ref{prop:inequality0}~and~\ref{prop:inequality} in the second inequality.

Let $g_\Mean(x) = (y_1,y_2)$ where $y_1,y_2\in [0,1]^{\nicefrac{n}{2}}$.
Then, we have 
$\| y_1 - x_1 \| 
= \nicefrac{1}{2}\cdot \| f_A(x) - x_1 \|$ and  $\| y_2 - x_2 \| = \nicefrac{1}{2}\cdot \| f_B(x) - x_2 \|$.
Hence if $\|g_\Mean(x)-x\| \leq \eps$ then $\|f_\Concat(x)-x\| \leq 2\eps$.
\end{proof}

\begin{lemma}
Let $n\in\N$  and $\lambda_A,\lambda_B,\eps\ge 0$. It holds that 
\[ \CC(\Mean_{n,\eps,\lambda_A,\lambda_B}) \leq \CC(\Comp_{n,\eps,\frac{\lambda_A}{2}+1,\lambda_B+2}). \]
\end{lemma}

\begin{proof}
Player $A$ gets $f_A:[0,1]^n\rightarrow [0,1]^n$.
Player $B$ gets $f_B:[0,1]^n\rightarrow [0,1]^n$.
Define \mbox{${g}_A:[0,1]^n\rightarrow [0,1]^{2n}$} for every $x\in [0,1]^n$ as
\[{g}_A(x) = \left( \frac{1}{2}\cdot f_A(x),x\right) ,\]
and define ${g}_B:[0,1]^{2n}\rightarrow [0,1]^n$ for every $x_1,x_2\in [0,1]^{n}$ as
\[{g}_B(x_1,x_2) =  x_1 + \frac{1}{2}\cdot f_B(x_2). \]

We now show that if the Lipschitz constant of $f_A$ is $\lambda_A$ then $g_A$ is at most \mbox{$\left(\frac{\lambda_A}{2}+1\right)$-Lipschitz}.
\begin{align*}
\|g_A(x)-g_A(y)\|&= \left\|\left(\frac{1}{2}\cdot\left(f_A(x)-f_A(y)\right),x-y\right)\right\|\\ 
&\le \left(\frac{1}{2}\cdot\left\|f_A(x)-f_A(y)\right\|\right)+\left\|x-y\right\|\\
&\le \left(\frac{\lambda_A}{2}+1\right)\cdot \left\|x-y\right\| ,
\end{align*}
where we used Proposition~\ref{prop:inequality0} in the first inequality.

\begin{sloppypar}Similarly, we show that if the Lipschitz constant of $f_B$ is $\lambda_B$ then $g_B$ is at most \mbox{$\left(\frac{\lambda_B}{2}+1\right)$-Lipschitz}.
\begin{align*}
\|g_B(x_1,x_2)-g_B(y_1,y_2)\|&= \left\|x_1 -y_1 + \frac{1}{2}\cdot \left(f_B(x_2) - f_B(y_2)\right)\right\|\\
&\le \left\|x_1 -y_1\right\| + \left\|\frac{1}{2}\cdot \left(f_B(x_2) - f_B(y_2)\right)\right\|\\
&\le \left\|x_1 -y_1\right\| + \frac{\lambda_B}{2}\cdot\left\| x_2-y_2\right\|\\
&\le 2\cdot \left\|x_1 -y_1,x_2-y_2\right\| + \lambda_B\cdot\left\| x_1-y_1,x_2-y_2\right\|\\
&\le \left(\lambda_B+2\right)\cdot \left\|x_1 -y_1,x_2-y_2\right\|. 
\end{align*}\end{sloppypar}

Finally, notice that for every $x\in[0,1]^n$ we have $f_{\Mean}(x)=g_{\Comp}(x)$, and the lemma follows.
\end{proof}

\begin{lemma}
Let $n\in\N$  and $\lambda_A,\lambda_B,\eps\ge 0$. It holds that
\[ \CC(\Comp_{n,\eps,\lambda_A,\lambda_B}) \leq \CC(\Concat_{O(n),O_{\lambda_B}(\eps),4(\lambda_A+1),4(\lambda_B+1)}). \]
\end{lemma}

\begin{proof}
Player $A$ gets $f_A:[0,1]^n\rightarrow[0,1]^m$ and player $B$ gets $f_B:[0,1]^m\rightarrow[0,1]^n$, where $m=O(n)$.
Define $g_A,g_B:[0,1]^{2(n+m)}\rightarrow [0,1]^{n+m}$ for every $a,x_1\in [0,1]^n$ and $b,x_2\in [0,1]^m$ as
\[ g_A(a,x_1,b,x_2) = (a,f_A(x_2)) \text{    and    }  g_B(a,x_1,b,x_2) = (b,f_B(x_1)).  \]

Let $a_x,x_1,a_y,y_1\in [0,1]^n$ and $b_x,x_2,b_y,y_2\in [0,1]^m$ and denote $x=(a_x,x_1,b_x,x_2)$, $y=(a_y,y_1,b_y,y_2)$.
First, we check the Lipschitz constant of $g_A$:
\begin{align*}
\| g_A(x) - g_A(y) \| &= \|(a_x,f_A(x_2))-(a_y,f_A(y_2))\|\\
&\leq \|a_x-a_y\| + \lambda_A\|x_2-y_2\|\\
&\leq 4(\lambda_A+1)\cdot \|x-y\|,
\end{align*} 
where the last inequality follows from Proposition~\ref{prop:inequality0}.

Similarly, we check the Lipschitz constant of $g_B$:
\begin{align*}
\| g_B(x) - g_B(y) \| &= \| (b_x,f_B(x_1))-(b_y,f_B(y_1))\|\\
&\leq  \|b_x-b_y\| + \lambda_B\|x_1-y_1\| \\
&\leq 4(\lambda_B+1)\cdot \|x-y\|,\end{align*} 
where the last inequality follows from Proposition~\ref{prop:inequality0}.

Next, we check the approximation factor we get for $f_\Comp$:
Let $c$ be a constant larger than $\nicefrac{m}{n}$ and $\nicefrac{n}{m}$.
Assume $\|g_\Concat(x)-x\| \leq \eps$. Then 
\[ \|g_A(x)-(a_x,x_1)\| = \|(a,f_A(x_2)) - (a_x,x_1)\|\leq 2\eps \]
 and 
 \[ \|f_A(x_2) - x_1 \| \leq 2\eps\cdot\frac{n+m}{m} \leq 2\eps (1+c).  \]
Similarly,  $\|g_B(x)-(b_x,x_2)\|\leq 2\eps$.
\[ \|g_B(x)-(b_x,x_2)\| = \|(b,f_B(x_1)) - (b_x,x_2)\|\leq 2\eps \]
 and 
 \[ \|f_B(x_1) - x_2 \| \leq 2\eps\cdot\frac{n+m}{n} \leq 2\eps(1+c).  \]
We get that
\begin{align*}
\| f_B(f_A(x_2)) - x_2 \|   &\leq   \| f_B(f_A(x_2)) - f_B(x_1) \| + \| f_B(x_1) - x_2 \| \\
&\leq \| f_B(f_A(x_2)) - f_B(x_1) \| + 2\eps(1+c) \\
&\leq \lambda_B \| f_A(x_2) - x_1  \| + 2\eps(1+c) \\
&\leq  2\eps(1+c)(\lambda_B+1)    \qedhere
\end{align*}
\end{proof}

\section{Lower Bound on Brouwer Problems}\label{sec:lowerbound}

In this section, we show an exponential lower bound in the dimension on the three Brouwer problems introduced in Section~\ref{sec:problems} (i.e., a polynomial lower bound in the size of the inputs ).
We begin by introducing the Local Brouwer problem, and then recall the lower bound of \cite{BR16} for Local Brouwer problem, and finally prove Theorem~\ref{thm:main} (and consequently Corollary~\ref{cor:three_problems_lb}).

Let $n,r,N\in\N$. Let $\mathcal{P}(N,r)$ denote the set of all subsets of size $r$ over the universe $[N]$.  
Assume that every $(x,y)\in\{0,1\}^N\times\{0,1\}^N$ defines 
a function $f_{x,y}:[0,1]^n\rightarrow [0,1]^n$.
We say that a collection of functions $\left\{f_{x,y}\right\}_{x,y\in\{0,1\}^N}$ is \emph{$r$-local} if there exist functions 
$\L:[0,1]^n\to \mathcal{P}(N,r)$ and $f':\{0,1\}^{2r}\times[0,1]^n\rightarrow [0,1]^n$ such that  for all $(x,y)\in \{0,1\}^N\times\{0,1\}^N$ and $z\in[0,1]^n$, we have:  
\begin{align}
f_{x,y}(z) = f'(x\vert_{\L(z)},y\vert_{\L(z)},z).\label{eq:local}
\end{align}

Informally,  for every point $z$ in $[0,1]^n$, $z$ decides through some fixed function $\L$ (independent of $x$ and $y$), as to which bits of $x$ and $y$ are relevant to compute $f_{x,y}(z)$.

The following is the formal definition of the Local Brouwer problem for two players.
We denote its randomized communication complexity 
$\CC(\Local_{p,n,\eps,\lambda, r})$. 

\begin{definition}[Local Brouwer Problem]\label{def:local}
Let $p\in\norm$, $n,r,N\in\N$  such that $n=\Theta(\log N)$, $\lambda\ge 0$, $\eps \geq 0$. Let $\L:[0,1]^n\to \mathcal{P}(N,r)$
and $f':\{0,1\}^{2r}\times[0,1]^n\rightarrow [0,1]^n$, and $\left\{f_{x,y}\right\}_{x,y\in\{0,1\}^N}$ be a collection of functions that are $r$-local with respect to $\L$ and $f'$ and each function in the collection is $\lambda$-Lipschitz continuous.
The \emph{Local Brouwer problem}
for two players $A$ and $B$ is as follows.
Let $p,n,r,N,\lambda,\eps \geq 0$ and $\L,f'$ be publicly known parameters.
Player $A$ is given $x\in\{0,1\}^N$ as input, and $y\in\{0,1\}^N$ is given to player $B$ as input.
Their goal is to output $z\in [0,1]^n$ such that 
$ \|f_{x,y}(z) - z\|_p \leq \eps $.
\end{definition}
 
In the case of the reductions in \cite{BR16} and \cite{GR18}, $r$ is the size of the inputs of each player to the gadgets used in the simulation theorem. We are interested in the constant gadget size simulation theorem of \cite{GP14a} used in \cite{GR18}, but the embedding in the Euclidean norm described in \cite{BR16} (which is essentially the embedding given in \cite{R16}) suffices for us. We note here that for the max norm, we can use the embedding given in \cite{HPV89} (simplified in \cite{B16,R15}).  Thus we have,
 
\begin{theorem}[\cite{HPV89,R16,BR16,GR18}]\label{thm:BR17}
Let $p\in \{2,\infty\}$ and $n,N\in\mathbb{N}$ such that $n=\Theta(\log N)$.
There exist constants $\varepsilon_0>0,\lambda_0>1, r_0>0,$ and $\L:[0,1]^n\to \mathcal{P}(N,r_0)$
and $f':\{0,1\}^{2r_0}\times[0,1]^n\rightarrow [0,1]^n$ such that  the following holds 
\[\CC(\Local_{p,n,\varepsilon_0,\lambda_0,r_0})=2^{\Omega(n)}.\]
\end{theorem}

The proof of Theorem~\ref{thm:main} follows from Theorem~\ref{thm:BR17}  above and the following theorem.  
Now we state and prove the main result of this paper.
  
\begin{theorem}\label{thm:loc2comp}
Let $p\in\{2,\infty\}$, $n,N,r\in\N$ such that $n=\Theta(\log N)$,  $\lambda\ge 1$, $\eps \geq 0$, let $\L:[0,1]^n\to \mathcal{P}(N,r)$
and $f':\{0,1\}^{2r}\times[0,1]^n\rightarrow [0,1]^n$ .
Then
\[ \CC(\Local_{p,n,\eps,\lambda,r}) \leq \CC(\Comp_{p,n,\eps,\lambda,(2^r+1)^{1/p}\cdot (\lambda+1)}).  \]
\end{theorem}

\begin{proof}
Let $p,n,r,\lambda, \L,f'$, and $\eps \geq 0$ be publicly known parameters.
Player $A$ gets $x$ and player $B$ gets $y$.
Given $x$, player $A$ defines the function $f_A:[0,1]^n\rightarrow [0,1]^{n(2^r+1)}$
for every $z\in[0,1]^n$ as
\[ f_A(z) = \left(f'(x\vert_{\L(z)},\beta_1,z),f'(x\vert_{\L(z)},\beta_2,z),\ldots,f'(x\vert_{\L(z)},\beta_{2^r},z),z\right), \]
where $\beta_1,\beta_2,\ldots,\beta_{2^r}$ is the enumeration of the elements of  $\{0,1\}^r$ in some canonical ordering.

Given $y$, player $B$ defines the function $f_B:[0,1]^{n(2^r+1)}\rightarrow [0,1]^n$. Before describing $f_B$, we define  $\mathcal{C}_t\subseteq [0,1]^{n(2^r+1)}$ for every $t\in\{0,1\}^N$ as
$$
\mathcal{C}_t=\left\{\left(f'(t\vert_{\L(z)},\beta_1,z),f'(t\vert_{\L(z)},\beta_2,z),\ldots,f'(t\vert_{\L(z)},\beta_{2^r},z),z\right)\mid \ z\in [0,1]^{n}\right\}.
$$

Define $\mathcal{C}\subseteq [0,1]^{n(2^r+1)}$ as 
$$
\mathcal{C}=\underset{t\in\{0,1\}^N}{\bigcup}\ \mathcal{C}_t.
$$

Player $B$ first defines a function $g_B:\mathcal{C}\rightarrow [0,1]^n$ as follows.
For every $(w_1,w_2,\ldots,w_{2^r},z)\in\mathcal{C}$ define 
$ g_B((w_1,w_2,\ldots,w_{2^r},z)) = w_i $, where $i$ is the index (according to the fixed ordering of elements of $\{0,1\}^r$) such that $\beta_i=y\vert_{\L(z)}$. Finally, we define ${f}_B:[0,1]^{n(2^r+1)}\rightarrow [0,1]^n$ using Lemmas~\ref{lem:extensionMax}~and~\ref{lem:extensionEuclid} as an extension of $g_B$ to $[0,1]^{n(2^r+1)}$.

Notice that the range of $f_A$ is contained in $\mathcal{C}$ (in fact in $\mathcal{C}_x$), and therefore $f_B\circ f_A=g_B\circ f_A$.
Then for every $z\in[0,1]^n$,
\begin{align*}
f_B(f_A(z))=g_B(f_A(z)) &= g_B\left(f'(x\vert_{\L(z)},\beta_1,z),f'(x\vert_{\L(z)},\beta_2,z),\ldots,f'(x\vert_{\L(z)},\beta_{2^r},z),z\right)\\
&= f'(x\vert_{\L(z)},y\vert_{\L(z)},z) = f_{x,y}(z) .
\end{align*}
Hence the composed function $f_B(f_A(\cdot))$ and $f_{x,y}$ have the same Lipschitz constant 
and the same approximate fixed points over $[0,1]^n$.
\allowdisplaybreaks

All that is left to prove are bounds on the Lipschitz constants of $f_A$ and $f_B$. Below we show that $f_A$ is $\lambda$-Lipschitz.
\begin{align*}
\left\|f_A(z)-f_A(z')\right\|_p^p&=\|\left(f'(x\vert_{\L(z)},\beta_1,z),f'(x\vert_{\L(z)},\beta_2,z),\ldots,f'(x\vert_{\L(z)},\beta_{2^r},z),z\right)\\
&\phantom{jjksdh}-\left(f'(x\vert_{\L(z')},\beta_1,z'),f'(x\vert_{\L(z')},\beta_2,z'),\ldots,f'(x\vert_{\L(z')},\beta_{2^r},z'),z'\right)\|_p^p\\
&= \|\left(f'(x\vert_{\L(z)},\beta_1,z)-f'(x\vert_{\L(z')},\beta_1,z'),f'(x\vert_{\L(z)},\beta_2,z)\right.\\
&\phantom{hk}\left.-f'(x\vert_{\L(z')},\beta_2,z'),\ldots ,f'(x\vert_{\L(z)},\beta_{2^r},z)-f'(x\vert_{\L(z')},\beta_{2^r},z'),z-z'\right)\|_p^p\\
&= 	\frac{1}{2^r+1}\cdot \left(\|f'(x\vert_{\L(z)},\beta_1,z)-f'(x\vert_{\L(z')},\beta_1,z')\|_p^p\right.\\
&\phantom{sfjlsjfldflkjsfls}\left.+\|f'(x\vert_{\L(z)},\beta_2,z)-f'(x\vert_{\L(z')},\beta_2,z')\|_p^p+\cdots\right.\\
&\phantom{sdkkkljlkjljhk}\left.\cdots +\|f'(x\vert_{\L(z)},\beta_{2^r},z)-f'(x\vert_{\L(z')},\beta_{2^r},z')\|_p^p+\|z-z'\|_p^p\right)\\
&\le \frac{2^r\lambda^p+1}{2^r+1}\cdot \|z-z'\|_p^p\\
&\le \lambda^p\cdot \|z-z'\|_p^p.
\end{align*}

Finally, we show below that $g_B$ is $O_r(\lambda)$-Lipschitz. This implies that $f_B$ is $O_r(\lambda)$-Lipschitz.
\begin{align*}
&\|g_B\left(f'(t\vert_{\L(z)},\beta_1,z),f'(t\vert_{\L(z)},\beta_2,z),\ldots,f'(t\vert_{\L(z)},\beta_{2^r},z),z\right)\\
&\phantom{jdsfksdh}-g_B\left(f'(t'\vert_{\L(z')},\beta_1,z'),f'(t'\vert_{\L(z')},\beta_2,z'),\ldots,f'(t'\vert_{\L(z')},\beta_{2^r},z'),z'\right)\|_p\\
&=\left\|f'(t\vert_{\L(z)},y\vert_{\L(z)},z)-f'(t'\vert_{\L(z')},y\vert_{\L(z')},z')\right\|_p \\
&\le \left\|f'(t\vert_{\L(z)},y\vert_{\L(z)},z)-f'(t'\vert_{\L(z)},y\vert_{\L(z)},z)\right\|_p  \\
&\phantom{jksdnfjks}+ \left\|f'(t'\vert_{\L(z)},y\vert_{\L(z)},z)-f'(t'\vert_{\L(z')},y\vert_{\L(z')},z')\right\|_p\\
&\le \left\|f'(t\vert_{\L(z)},\beta_i,z)-f'(t'\vert_{\L(z)},\beta_i,z)\right\|_p + \lambda\left\|z-z'\right\|_p\ \ \ \ \ \ \text{(where $\beta_i:=y\vert_{\L(z)}$)}\\
&\le (\lambda+1)\cdot \left\|f'(t\vert_{\L(z)},\beta_i,z)-f'(t'\vert_{\L(z)},\beta_i,z)\right\|_p + \left(\lambda+1\right)\cdot\left\|z-z'\right\|_p\\
&\le (2^r+1)^{1/p}\cdot (\lambda+1)\cdot \|\left(f'(t\vert_{\L(z)},\beta_1,z),f'(t\vert_{\L(z)},\beta_2,z),\ldots,f'(t\vert_{\L(z)},\beta_{2^r},z),z\right)\\
&\phantom{jdsfksdh}-\left(f'(t'\vert_{\L(z')},\beta_1,z'),f'(t'\vert_{\L(z')},\beta_2,z'),\ldots,f'(t'\vert_{\L(z')},\beta_{2^r},z'),z'\right)\|_p,
\end{align*}
where the last inequality follows from Proposition~\ref{prop:inequality}.
\end{proof}

Corollary~\ref{cor:three_problems_lb} follows from Theorems~\ref{thm:three_problems} and \ref{thm:main}, as the Lipschitz constants of $f_A$ and $f_B$ in the proof above are $O(\lambda)$.
Finally, we note that it might be possible to extend the embedding given in \cite{R16} to all $\ell_p$-norms (in a straightforward manner), in which case if we have extension theorem for the domain $\mathcal{C}$ in the above proof for other $\ell_p$-norms (see related discussion in Section~\ref{sec:total}) then we would obtain the lower bound in Theorem~\ref{thm:main} for all $\ell_p$ norms.

An interesting open problem is to extend our lower bounds to the multiparty communication model. More formally, consider the $k$-party Composition Brouwer problem (denoted by $k-\Comp_{p,n,\varepsilon,\lambda_1,\ldots ,\lambda_k}$),  where for every $i\in[k]$, Player $i$ gets  a $\lambda_i$-Lipschitz function $f_i:[0,1]^n\rightarrow [0,1]^n$  and their goal is to output $x\in [0,1]^n$ such that 
$ \|f_{k}(f_{k-1}(\cdots f_1(x)\cdots )) - x\|_p \leq \eps $.
Naively, we can simply compute the value of the composed function on a dense enough grid of $[0,1]^n$ and obtain a protocol for $k-\Comp_{p,n,\varepsilon,\lambda_1,\ldots ,\lambda_k}$ with $\lambda^{O(nk)}$ bits of communication, where $\lambda=\underset{i\in[k]}{\max}\ \lambda_i$ and $\varepsilon$ is some small constant (similar to the proof of Lemma~\ref{lem:upper}). 
 Following the proof of Theorem~\ref{thm:main} it is also easy to show that $\CC(k-\Comp_{p,n,\varepsilon,\lambda_1,\ldots ,\lambda_k})\ge 2^{\Omega(n)} + k$. Can we obtain stronger lower bounds for this problem?

\begin{open}
What is the randomized communication complexity of $k-\Comp_{p,n,\varepsilon,\lambda_1,\ldots ,\lambda_k}$?
\end{open}

\section{Concatenation Sperner Problem}\label{sec:sperner}

We begin the section by formalizing the notion of a Sperner-coloring.

\begin{definition}[Sperner-coloring]
	A $(d+1)$-coloring $c$ of a triangulated, $d$-dimensional simplex $\Delta=\conv(v_0,\ldots,v_{d})$ is a \emph{Sperner-coloring} if $c(v_i)=i$ and every vertex $x$ gets the color of one of the vertices of the smallest face of $\Delta$ that contains $x$.  
\end{definition}

We say that a full-dimensional face of the triangulation (which is a small simplex) is \emph{panchromatic} if all its vertices have a different color, i.e., all the $d+1$ colors appear.
Sperner's Lemma asserts that in every Sperner-coloring there are an odd number of panchromatic simplices\footnote{For the discussion in this paper, we only use the fact that every Sperner-coloring implies the existence of at least one panchromatic simplex.}. For every color $i$, its color class is the set of all points in the triangulation that are colored $i$ by the Sperner-coloring.	

The natural Sperner problem we associate with Sperner's Lemma is given a Sperner-coloring of a fixed triangulation of  $\Delta$, find a panchromatic simplex. There are many interesting realizations of this Sperner problem in the communication model. In this paper, we consider the following two-player communication problem.

\begin{definition}[Concatenation Sperner Problem ($\Sp_{d,n}^t$)]
Let   $n,d,t\in\N$ such that $t\le d<n$.
The \emph{Concatenation Sperner problem} 
for two players $A$ and $B$ is as follows.
Let $n,d$, and a triangulation of $n$ points of the $d$-simplex labeled by $[n]$ be publicly known parameters.
Player $A$ gets $t$ disjoint subsets $C_0,\ldots ,C_{t-1} \subset [n]$ corresponding to the color classes of the first $t$ colors of the Sperner-coloring.
Player $B$ gets $d-t+1$ disjoint subsets $C_{t},\ldots ,C_{d} \subset [n]$ corresponding to the color class of the last $d-t+1$ colors of the Sperner-coloring.
Their goal is to output an $x\in C_0\times \cdots \times C_{d}$ that is a simplex in the triangulation (or show that any of the above conditions are not satisfied).
\end{definition}

We denote the randomized (resp.\ deterministic) communication complexity of this problem by $\CC(\Sp_{d,n}^t)$ (resp.\ $\Det(\Sp_{d,n}^t)$) for an $n$-vertex triangulation of a $d$-dimensional simplex where Player $A$ gets $t$ color classes and Player $B$ gets the remaining color classes. An easy observation when one of the players gets a single color class is stated below.

\begin{remark}\label{rem:sp0}
$\Det(\Sp^{d}_{d,n})=O(\log n)$
\end{remark} 

The above remark follows from the observation that if Player $A$ has all but one color, she knows that every vertex at Player $B$ has color $d+1$, thus she can alone output a panchromatic simplex.
We prove that the problem can be solved efficiently even if $A$ has all but two colors.
 
\begin{theorem}[Restating Theorem~\ref{thm:d4}]\label{thm:14}$\Det(\Sp^{d-1}_{d,n})=O(\log^2 n)$.
\end{theorem}

The above result is more interesting for small $d$.

\begin{corollary}
\label{cor:d4}	
For all $t\in \{1,\ldots ,d\}$, we have	 $\Det(\Sp_{d,n}^t)=O(\log^2 n)$ when $d\le 4$. 
\end{corollary}

\begin{sloppypar}The proof of the above corollary follows from Theorem~\ref{thm:14} and Remark~\ref{rem:sp0} as $\Sp^{t}_{d,n}=~\Sp^{d+1-t}_{d,n}$. 
The heart of the proof of Theorem~\ref{thm:14} is the following variant of Sperner's Lemma.\end{sloppypar}

\begin{definition}[Surplus Sperner-coloring]
	A $d$-coloring $c$ of a triangulated, $d$-dimensional simplex $\Delta=\conv(v_0,\ldots,v_{d})$ is a \emph{surplus Sperner-coloring} if $c(v_i)=i$ for $i<d$ and $c(v_d)=0$, and every vertex $x$ gets the color of one of the vertices of the smallest face of $\Delta$ that contains $x$.
\end{definition}

Define a graph $G$ whose vertices are the full-dimensional faces (small simplices) of a surplus Sperner-colored triangulated $d$-simplex $\Delta$, and two vertices are connected by an edge if they share a panchromatic facet, i.e., a facet whose $d$ vertices contain all $d$ colors.
Denote the facet of $\Delta$ avoiding $v_d$ by $F_0$, and the facet avoiding $v_0$ by $F_d$.
Add two more vertices to the graph, $f_0$ and $f_d$, such that $f_0$ (resp.\ $f_d$) is connected to all small simplices with a panchromatic facet on $F_0$ (resp.\ $F_d$).
By applying the $(d-1)$-dimensional Sperner's Lemma to $F_0$ (resp.\ $F_d$), we can conclude that there are an odd number of small simplices with a panchromatic facet on $F_0$ (resp.\ $F_d$), thus the degrees of $f_0$ and $f_d$ are both odd.
But since $G$ consists of disjoint paths and cycles, this implies that there is a path between $f_0$ and $f_d$ 
(see Figure \ref{fig:surplussperner}).
Therefore we have proved the following.

\begin{lemma}[Surplus Sperner Lemma]
	There is a chain of small panchromatic simplices between $F_0$ and $F_d$ in any surplus Sperner-colored $d$-simplex, such that the neighboring simplices in the chain always share a panchromatic facet, and the two facets that fall on $F_0$ and $F_d$ are also panchromatic.
\end{lemma}

\begin{figure}[!h]
\begin{center}
	\includegraphics[width=\linewidth]{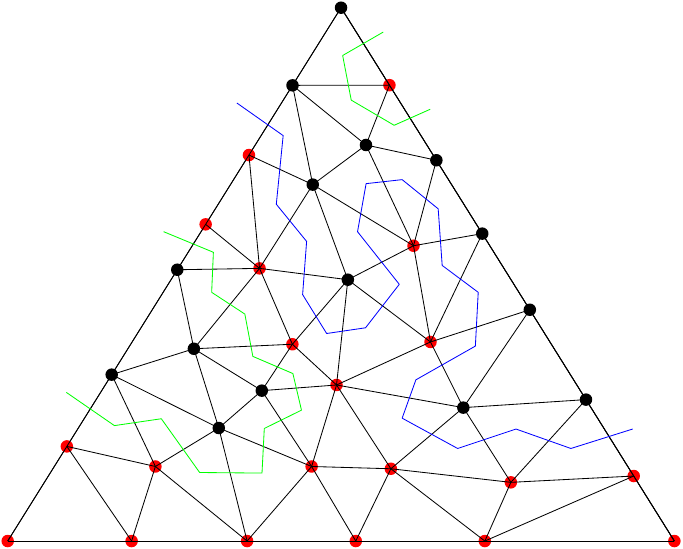}
	\caption{Depiction of a surplus Sperner-coloring of a triangle (i.e., $d=2$), where \emph{black} stands for color $1$ and \emph{red} for color $0$. The blue polygonal line marks a	chain of triangles between two sides, as guaranteed to exist by the Surplus Sperner Lemma. The green polygonal lines show other paths constructed in the proof.}
	\label{fig:surplussperner}
\end{center}
\end{figure}

Now it is easy to establish the proof of Theorem \ref{thm:14}.
\begin{proof}[Proof of Theorem~\ref{thm:14}]
The players follow the below protocol. 
With a relabeling, suppose that the two colors missing from $A$ are $0$ and $d$.
Define the surplus Sperner-coloring $c'$ as $c'=c \bmod d$, i.e., for all $i\in[n]$, we have $c'(i)=c(i)$, except when $c(i)=d$, in which case we set $c'(i)=0$.

\begin{enumerate}
\item Player $A$ builds the graph $G$ described above (with zero bits of communication).
\item Let $P=p_1\cdots p_r$ be the path in $G$ guaranteed by the Surplus Sperner Lemma. Player $A$ would like to label an edge in $P$ by 0 (resp.\ $d$) if the common facet between the two vertices in the path has color 0 (resp.\ $d$). Player $A$ labels the outgoing edge of $p_1$ in $P$ by 0 and the outgoing edge of $p_r$ in $P$ by $d$ with no communication. 
\item The players communicate by a binary search method until they find a vertex whose two outgoing edges are both labeled and have different labels. Such a vertex corresponds to a panchromatic simplex in the triangulation.
\end{enumerate}

It is clear that there are $\log r=O(\log n)$ rounds of communication and in each round there are $O(\log n)$ bits of communication.
\end{proof}

We can adopt the protocol from the proof of Theorem~\ref{thm:14} to obtain the following slightly stronger result.

\begin{corollary}\label{cor:3-Sp}
Let   $n\in\N$. Consider the three-party communication problem in the broadcast model of finding a \emph{panchromatic triangle of a concatenation Sperner-coloring} 
for three players $1$, $2$, and $3$.  
Let $n$, and a triangulation of $n$ points of the triangle labeled by $[n]$ be publicly known parameters.
For all $i\in [3]$,
Player $i$ gets a subset $C_{i-1} \subset [n]$ corresponding to the color class of color ${i-1}$ of the Sperner-coloring.
Their goal is to output an $x\in C_0\times C_1 \times C_2$ that is a triangle in the triangulation (or show that  their sets are not mutually disjoint and exhaustive). Then there is a deterministic protocol for this problem with $O(\log^2 n)$ bits of  communication.
\end{corollary}

The above result should be compared with its counterpart in the query model \cite{CS98} and Turing machine model \cite{CD09} which are both intractable even for the planar case. We would like to highlight that the construction of the graph $G$ is not feasible in the query model as it requires a large number of queries and is not feasible in the Turing machine model as it requires exponential time. 

It is also worth exploring if the upper bound on $\Det(\Sp^{d-1}_{d,n})$ can be improved to $O(\log n)$, at least in the case where we allow randomized protocols. This is discussed further in Section~\ref{sec:KW}.

Next we show that in higher dimensions the Concatenation Sperner problem is at least as hard as the Composition Brouwer problem. The proof of Theorem~\ref{thm:SpLB} follows from the lower bound in Theorem~\ref{thm:loc2comp} obtained via Theorem~\ref{thm:BR17}.

\begin{theorem}\label{thm:concSp2compBr}
	$\CC\left(\Sp_{d+1,(\lambda/\varepsilon)^{2d}}^{d/2}\right)=\Omega(\CC(\Comp_{2,(d-1)/2,\varepsilon,\lambda,O(\lambda)}))=2^{\Omega(d)}$. 
\end{theorem}

Before we prove Theorem \ref{thm:concSp2compBr}, we sketch how the reduction in the other way would go (for a special instance of Sperner, described below).
We do this to give some intuition which will be helpful later in understanding the proof.

\begin{figure}
	\begin{center}
		\includegraphics[width=\linewidth]{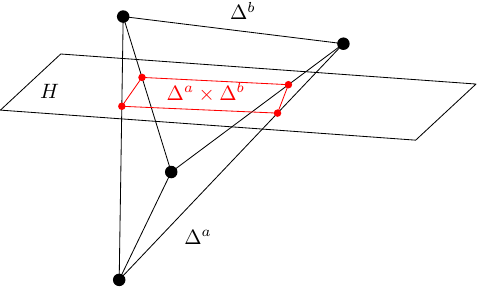}
		\caption{$\Delta^a \times \Delta^b$ as $H\cap \Delta^{a+b+1}$ for $a=b=1$.}
		\label{fig:tetra2square}
	\end{center}
\end{figure}

Let $a\le d$. Suppose that player $A$ holds the first $a+1$ color classes, which belong to $v_0,\ldots, v_{a}$ and player $B$ holds the remaining $b+1:=d-a$ color classes, which belong to $v_{a+1},\ldots ,v_d$.
We denote the convex hulls of these vertices by $\Delta^a$ and $\Delta^b$, respectively.
Consider the cross-section of $\Delta^d$ by a hyperplane $H$ that separates $\Delta^a$ from $\Delta^b$; see Figure \ref{fig:tetra2square}.
(We suppose that no vertex of the triangulation falls on $H$.)
Denote the halfspace that contains $\Delta^a$ as $H_A$ and the halfspace that contains $\Delta^b$ as $H_B$.
We suppose that all vertices of the triangulation in $H_A$ are colored with the first $a+1$ colors, and all vertices in $H_B$ are colored with the remaining $b+1$ colors.
This implies that every panchromatic simplex is intersected by $H$.
We say that a simplex is A-panchromatic, if it contains all of the first $a+1$ colors, and that a simplex is B-panchromatic, if it contains all of the remaining $b+1$ colors.
Thus, a simplex is panchromatic if it is both A-panchromatic and B-panchromatic; this can only happen for simplices that intersect $H$.
The points of the cross-section $\Delta^d \cap H$ are in a natural bijection with the points of $\Delta^a \times \Delta^b$, so from now on we will refer to each as a point $(p,q)\in \Delta^a \times \Delta^b$.

Now we show how they could solve this special case (i.e., when their colors are separated by $H$) using a protocol for the Composition Brouwer problem.
We extend the coloring from the vertices of the triangulation to all the points by coloring any point $p\in \Delta^d$ with a color of the vertices of the smallest face of the triangulation containing $p$ (keeping the rule that in $H_A$ everything is colored with $A$'s colors and in $H_B$ everything is colored with $B$'s colors).

For any point $q\in \Delta^b$, define a refinement of our given triangulation $\mathcal{T}$ by adding the intersection of the faces of $\mathcal{T}$ and the subspaces through $q$ and some of the vertices of $\Delta^a$ to the triangulation to obtain ${\mathcal{T}}(q)$.
Notice that the conditions of Sperner's Lemma hold for ${\mathcal{T}}(q)$, thus player $A$ knows an $f_A(q)\in \Delta^a$ for which $(f_A(q),q)\in \Delta^a \times \Delta^b$ is contained in an A-panchromatic simplex of ${\mathcal{T}}(q)$, and thus also in an A-panchromatic simplex of the original triangulation $\mathcal{T}$.
Similarly, for every $p\in \Delta^a$, $B$ knows a $f_B(p)\in \Delta^b$ for which $(p,f_B(p))\in \Delta^a \times \Delta^b$ is contained in a B-panchromatic simplex.
Therefore, the players hold two continuous functions, $f_A: \Delta^b \to \Delta^a$ and $f_B: \Delta^a \to \Delta^b$, respectively, with the above properties.
If these functions were continuous\footnote{These functions are not always continuous, partly because we did not insist on extending the coloring in a nice way, but more importantly because there inherently might be multiple solutions, i.e., $f_A(q)$ can take several values from $\Delta^a$. If we replaced continuity by demanding the graph of the relation $f_A$ to be a $b$-manifold in $\Delta^a \times \Delta^b$, then we could achieve the reverse direction of the reduction, but since we do not need this claim, we will not elaborate further.}, then using the protocol for Composition Brouwer problem, they could find a fixed point $q\in \Delta^b$ of $f_B\circ f_A$.
But then the simplex that contains $(f_A(q),q)$ is both A-panchromatic and B-panchromatic, thus panchromatic.

To prove Theorem \ref{thm:concSp2compBr}, we need the opposite of this argument, thus instead of simulating colorings by functions, we need to simulate functions by colorings.
This is captured by the following lemma.

\begin{lemma}\label{ftocol}
	 For any $\lambda$-Lipschitz function $f_A:\Delta^b\rightarrow \Delta^a$ in the Euclidean norm and a fine enough triangulation of $\Delta^d$, there is a coloring of the vertices in the triangulation in $H_A\subset \Delta^d$ as above such that if $(p,q)\in \Delta^a\times \Delta^b$ is in an A-panchromatic simplex, then $p$ is in a small neighborhood of $f_A(q)$.
\end{lemma}
\begin{proof}
	It is clearly enough to define our coloring on the vertices of the simplices that intersect $H$; the remaining vertices can be colored arbitrarily (respecting the boundary conditions required by the Sperner-coloring).
	
	First, we define an $(a+1)$-coloring $c_H$ on $\Delta^a \times \Delta^b$.
	The color of a point $(p,q) \in \Delta^a\times \Delta^b$ is defined as follows.
	Express $f_A(q)-p$ as a conical combination of the vertices of $\Delta^a$, i.e., $f_A(q)-p=\sum_{i=0}^a \mu_i (v_i-o_A)$, where $\mu_i\ge 0$ and $o_A$ denotes the center of $\Delta^a$.
	Note that in such a combination some $\mu_i=0$.
	Color $(p,q)$ with one such color, i.e., to a color $i$ whose coefficient $\mu_i=0$;	
	in case of $(f_A(q),q)$, color it arbitrarily (always respecting the boundary conditions required by the Sperner-coloring). 
	
	Suppose that all $a+1$ colors occur in an $\eps$-neighborhood of some $(p,q)$.
	Let $j=\arg\max \mu_i$.
	Since $(p,q)$'s color is not $j$, there is a $(p',q')$ for which in $f_A(q')-p'=\sum_{i=0}^a \mu_i' (v_i-o_A)$ we have $\mu_j'=0$.	
	We have $\|(f_A(q')-p')-(f_A(q)-p)\|_2\ge \|f_A(q')-f_A(q)\|_2 - \|p'-p\|_2 \ge \mu_j/d-\eps$, since $p'$ is in an $\eps$-neighborhood of $p$.
	But using that $f_A$ is $\lambda$-Lipschitz, we should have $\|f_A(q')-f_A(q)\|\le \lambda\|q'-q\|\le \lambda\eps$, since $q'$ is in an $\eps$-neighborhood of $q$.
	Putting these together, $\mu_j$ is small, thus $f_A(q)-p$ is small, thus $p$ is in a small neighborhood of $f_A(q)$ (where by small we mean that the volume of the neighborhood is $(\varepsilon/\lambda)^{O(d)}$).
	
	From $c_H$ we obtain a coloring $c_A$ of the vertices of the simplices that intersect $H$.
	Simply color each vertex to any color that occurs in it's simplex in $c_H$ (respecting the boundary conditions required by the Sperner-coloring).
	If a simplex is A-panchromatic, then all $a+1$ colors occur in it, thus for one of its points $(p,q)$, $p$ is in a small neighborhood of $f_A(q)$.
\end{proof}

\begin{proof}[Proof of Theorem \ref{thm:concSp2compBr}]
Starting from two functions $f_A: \Delta^b \to \Delta^a$ and $f_B: \Delta^a \to \Delta^b$ (where now $a=b=(d-1)/2$), we create a Sperner-coloring of $\Delta^d$.
It will have a hyperplane $H$ separating the colors of $A$ and $B$, as described above. 
Using Lemma \ref{ftocol}, we convert our continuous functionss $f_A: \Delta^b \to \Delta^a$ and $f_B: \Delta^a \to \Delta^b$ into Sperner-colorings.
Therefore, any protocol for the Concatenation Sperner problem for such special instances will also solve the Composition Brouwer problem. Finally note that we can easily extend the domain of the composition Brouwer functions from $[0,1]^{(d-1)/2}$ to $\Delta^{(d+1)/2}$ by taking a small simplex encapsulating the hypercube and defining the function on the region in the simplex outside the hypercube to stay within the hypercube, and therefore not creating any new fixed points, preserving the Lipschitz continuity, and decreasing the lower bound only by a polynomial factor. 
\end{proof}

Finally, we remark that it would be another, quite natural version of Sperner to study the following problem.
	We are given \emph{two} colorings, $c_A$ and $c_B$, of the vertices of a triangulation of $\Delta^a \times \Delta^b=\conv(v_i\mid i \in \{0,\ldots,a\}) \times \conv(v_i\mid i \in \{a+1,\ldots,a+b+1\})$.
	Player $A$ holds $c_A$, which is an $(a+1)$-coloring such that if $p\in \conv(v_i\mid i\in I \subset\{0,\ldots,a\})$, then $c_A(p,q)\in I$.
	Player $B$ holds $c_B$, which is a $(b+1)$-coloring such that if $q\in \conv(v_i\mid i\in I \subset\{a+1,\ldots,a+b+1\})$, then $c_B(p,q)\in I$.
	Their goal is to output a simplex that contains all $a+b+2$ colors.
	It can be easily seen from our proofs that our bounds also apply for this problem.

\section{Conclusion}

In this article, we showed that the three natural fixed point computation tasks in the communication model, namely, the Composition Brouwer problem, the Concatenation Brouwer problem, and the Mean Brouwer problem, are all equivalent up to polynomial factors (Theorem~\ref{thm:three_problems}). Moreover, we showed a lower bound of $2^{\Omega(n)}$ for the above three problems when the inputs to these problems are constant Lipschitz functions in $n$-dimensional space (Theorem~\ref{thm:main}). 
Finally, we initiated the study of finding a panchromatic simplex in a Sperner-coloring of a triangulation in the two-player communication model. 	 Rather surprisingly,  we showed that the Sperner problem  can be solved with small amount of communication when  the dimension is less than 5 (Theorem~\ref{thm:d4}), but on the other hand, we also showed strong lower bounds on the randomized communication complexity in high dimensions (Theorem~\ref{thm:SpLB}).

A natural research direction to pursue is if the lower bound in Theorem~\ref{thm:SpLB} can be extended to the case when $d$ is a constant. Note that if $d$ is a constant, the lower bound of $2^{\Omega(d)}$  shown in Theorem~\ref{thm:SpLB} is essentially trivial. Therefore, we ask:

\begin{open}
For a fixed constant $d>4$, is $\CC(\Sp_{d,n}^{\lfloor \nicefrac{d}{2}\rfloor})=n^{\Omega(1)}$?
\end{open}

A natural approach to try to prove the above result is to reduce the lower bound of Composition Brouwer problem in constant dimensions to the Concatenation Sperner problem using Theorem~\ref{thm:concSp2compBr}. However, we do not know of any non-trivial lower bounds for Composition Brouwer problem in constant dimensions. Therefore, we ask:

\begin{open}
For fixed constants $n>1,\lambda_A,\lambda_B\ge 1$, is $\CC(\Comp_{p,n,\varepsilon,\lambda_A,\lambda_B})=(\nicefrac{1}{\varepsilon})^{\Omega(1)}$?
\end{open}

We remark here that strong lower bounds have been known for the Brouwer problem in constant dimensions in the query model \cite{CD08,CT07,CD09} and we wonder if any of the embedding techniques developed in those works can be used to resolve the above open question.


\addcontentsline{toc}{section}{\protect\numberline{}Acknowledgements}%
\section*{Acknowledgements}
We thank Itai Benjamini, Suryateja Gavva, Assaf Naor, Aviad Rubinstein, Gideon Schechtman, Uli Wagner, Omri Weinstein, and Eylon Yogev for very helpful conversations and the anonymous reviewers for their valuable feedback on an earlier version of this manuscript. In particular, Aviad suggested to us to drop the discretization parameter from the definition of the Brouwer problems. 

Anat Ganor  received funding from the European Research Council (ERC) under the European Union’s Horizon 2020 research and innovation programme (grant agreement No 740282).  Karthik~C.~S.\  was supported by Irit Dinur's ERC-CoG grant 772839.  D{\"{o}}m{\"{o}}t{\"{o}}r P{\'{a}}lv{\"{o}}lgyi was supported by the Lend\"ulet program of the Hungarian Academy of Sciences (MTA), under grant number LP2017-19/2017.

\addcontentsline{toc}{section}{\protect\numberline{}References}%
\bibliographystyle{alpha}
\bibliography{refs}

\begin{thebibliography}{GPW17}

\bibitem[AH00]{AH00}
Ron Aharoni and Penny Haxell.
\newblock Hall's theorem for hypergraphs.
\newblock {\em J. Graph Theory}, 35(2):83--88, 2000.

\bibitem[Bab16]{B16}
Yakov Babichenko.
\newblock Query complexity of approximate {Nash} equilibria.
\newblock {\em J. {ACM}}, 63(4):36, 2016.

\bibitem[BR17]{BR16}
Yakov Babichenko and Aviad Rubinstein.
\newblock Communication complexity of approximate {Nash} equilibria.
\newblock In {\em Proceedings of the 49th Annual {ACM} {SIGACT} Symposium on
  Theory of Computing, {STOC} 2017, Montreal, QC, Canada, June 19-23, 2017},
  pages 878--889, 2017.

\bibitem[Bro12]{B12}
L.E.J. Brouwer.
\newblock \"{U}ber abbildung von mannigfaltigkeiten.
\newblock {\em Mathematische Annalen}, 71:97--115, 1912.

\bibitem[CD08]{CD08}
Xi~Chen and Xiaotie Deng.
\newblock Matching algorithmic bounds for finding a brouwer fixed point.
\newblock {\em J. {ACM}}, 55(3):13:1--13:26, 2008.

\bibitem[CD09]{CD09}
Xi~Chen and Xiaotie Deng.
\newblock On the complexity of 2d discrete fixed point problem.
\newblock {\em Theor. Comput. Sci.}, 410(44):4448--4456, 2009.

\bibitem[CDT09]{CDT09}
Xi~Chen, Xiaotie Deng, and Shang{-}Hua Teng.
\newblock Settling the complexity of computing two-player {Nash} equilibria.
\newblock {\em J. {ACM}}, 56(3), 2009.

\bibitem[CS98]{CS98}
Pierluigi Crescenzi and Riccardo Silvestri.
\newblock Sperner's lemma and robust machines.
\newblock {\em Computational Complexity}, 7(2):163--173, 1998.

\bibitem[CT07]{CT07}
Xi~Chen and Shang{-}Hua Teng.
\newblock Paths beyond local search: {A} tight bound for randomized fixed-point
  computation.
\newblock In {\em 48th Annual {IEEE} Symposium on Foundations of Computer
  Science {(FOCS} 2007), October 20-23, 2007, Providence, RI, USA,
  Proceedings}, pages 124--134, 2007.

\bibitem[Dan06]{D06}
Stefan~S. Dantchev.
\newblock On the complexity of the {Sperner} lemma.
\newblock In {\em Logical Approaches to Computational Barriers, Second
  Conference on Computability in Europe, CiE 2006, Swansea, UK, June 30-July 5,
  2006, Proceedings}, pages 115--124, 2006.

\bibitem[DKL19]{DKL19}
Roee David, {Karthik {C. S.}}, and Bundit Laekhanukit.
\newblock On the complexity of closest pair via polar-pair of point-sets.
\newblock {\em {SIAM} J. Discrete Math.}, 33(1):509--527, 2019.

\bibitem[ET76]{ET76}
Shimon Even and Robert~Endre Tarjan.
\newblock A combinatorial problem which is complete in polynomial space.
\newblock {\em J. {ACM}}, 23(4):710--719, 1976.

\bibitem[EY10]{EY10}
Kousha Etessami and Mihalis Yannakakis.
\newblock On the complexity of {Nash} equilibria and other fixed points.
\newblock {\em {SIAM} J. Comput.}, 39(6):2531--2597, 2010.

\bibitem[FISV09]{FISV09}
Katalin Friedl, G{\'{a}}bor Ivanyos, Miklos Santha, and Yves~F. Verhoeven.
\newblock On the black-box complexity of {Sperner's} lemma.
\newblock {\em Theory Comput. Syst.}, 45(3):629--646, 2009.

\bibitem[Gal79]{G79}
David Gale.
\newblock The game of hex and {Brouwer} fixed-point theorem.
\newblock {\em The American Mathematical Monthly}, 86(10):818--827, 1979.

\bibitem[GK18]{GK18}
Anat Ganor and {Karthik {C. S.}}
\newblock Communication complexity of correlated equilibrium with small
  support.
\newblock In {\em Approximation, Randomization, and Combinatorial Optimization.
  Algorithms and Techniques, {APPROX/RANDOM} 2018, August 20-22, 2018 -
  Princeton, NJ, {USA}}, pages 12:1--12:16, 2018.

\bibitem[GP14]{GP14a}
Mika G{\"{o}}{\"{o}}s and Toniann Pitassi.
\newblock Communication lower bounds via critical block sensitivity.
\newblock In {\em Symposium on Theory of Computing, {STOC} 2014, New York, NY,
  USA, May 31 - June 03, 2014}, pages 847--856, 2014.

\bibitem[GPW17]{GPW17}
Mika G{\"{o}}{\"{o}}s, Toniann Pitassi, and Thomas Watson.
\newblock Query-to-communication lifting for bpp.
\newblock {\em Electronic Colloquium on Computational Complexity {(ECCC)}},
  2017.

\bibitem[GR18]{GR18}
Mika G{\"{o}}{\"{o}}s and Aviad Rubinstein.
\newblock Near-optimal communication lower bounds for approximate {Nash}
  equilibria.
\newblock In {\em FOCS}, 2018.

\bibitem[Gri01]{G01}
Michelangelo Grigni.
\newblock A {Sperner} lemma complete for {PPA}.
\newblock {\em Inf. Process. Lett.}, 77(5-6):255--259, 2001.

\bibitem[Hax11]{H11}
Penny Haxell.
\newblock On forming committees.
\newblock {\em Am. Math. Mon.}, 118(9):777--788, 2011.

\bibitem[HN12]{HN12}
Trinh Huynh and Jakob Nordstr{\"{o}}m.
\newblock On the virtue of succinct proofs: amplifying communication complexity
  hardness to time-space trade-offs in proof complexity.
\newblock In Howard~J. Karloff and Toniann Pitassi, editors, {\em Proceedings
  of the 44th Symposium on Theory of Computing Conference, {STOC} 2012, New
  York, NY, USA, May 19 - 22, 2012}, pages 233--248. {ACM}, 2012.

\bibitem[HPV89]{HPV89}
Michael~D. Hirsch, Christos~H. Papadimitriou, and Stephen~A. Vavasis.
\newblock Exponential lower bounds for finding {Brouwer} fix points.
\newblock {\em J. Complexity}, 5(4):379--416, 1989.

\bibitem[II99]{II99}
Tatsuro Ichiishi and Adam Idzik.
\newblock Equitable allocation of divisible goods.
\newblock {\em Journal of Mathematical Economics}, 32(4):389 -- 400, 1999.

\bibitem[JST11]{JST11}
Hossein Jowhari, Mert Sa\u{g}lam, and G{\'{a}}bor Tardos.
\newblock Tight bounds for lp samplers, finding duplicates in streams, and
  related problems.
\newblock In {\em Proceedings of the 30th {ACM} {SIGMOD-SIGACT-SIGART}
  Symposium on Principles of Database Systems, {PODS} 2011, June 12-16, 2011,
  Athens, Greece}, pages 49--58, 2011.

\bibitem[Kir34]{K34}
M.~Kirszbraun.
\newblock {\"Uber} die zusammenziehende und {Lipschitzsche} {Transformationen}.
\newblock {\em Fundamenta Mathematicae}, 22(1):77--108, 1934.

\bibitem[Kla17]{K17}
Erica Klarreich.
\newblock In game theory, no clear path to equilibrium, July 2017.
\newblock
  http://www.quantamagazine.org/in-game-theory-no-clear-path-to-equilibrium-20170718/
  {[Online; posted 18-July-2017]}.

\bibitem[Mat07]{M07}
Jiri Matousek.
\newblock {\em Using the {Borsuk-Ulam} Theorem: Lectures on Topological Methods
  in Combinatorics and Geometry}.
\newblock Springer Publishing Company, Incorporated, 2007.

\bibitem[Mei17]{M17}
Or~Meir.
\newblock An efficient randomized protocol for every {Karchmer-Wigderson}
  relation with two rounds.
\newblock {\em Electronic Colloquium on Computational Complexity {(ECCC)}},
  24:129, 2017.

\bibitem[MM16]{MM16}
Konstantin Makarychev and Yury Makarychev.
\newblock Metric extension operators, vertex sparsifiers and {Lipschitz}
  extendability.
\newblock {\em Israel Journal of Mathematics}, 212(2):913--959, May 2016.

\bibitem[Nao01]{N01}
Assaf Naor.
\newblock A phase transition phenomenon between the isometric and isomorphic
  extension problems for {H\"older} functions between {$L_p$} spaces.
\newblock {\em Mathematika}, 48(1-2):253–271, 2001.

\bibitem[Nas51]{N51}
J.F. Nash.
\newblock Non-cooperative games.
\newblock {\em Annals of Mathematics}, 54(2):286--295, 1951.

\bibitem[Nas52]{N52}
John Nash.
\newblock Some games and machines for playing them.
\newblock {\em Rand Corp. technical report D-1164}, 1952.

\bibitem[Pap94]{P94}
Christos~H. Papadimitriou.
\newblock On the complexity of the parity argument and other inefficient proofs
  of existence.
\newblock {\em J. Comput. Syst. Sci.}, 48(3):498--532, 1994.

\bibitem[RM99]{RM99}
Ran Raz and Pierre McKenzie.
\newblock Separation of the monotone {NC} hierarchy.
\newblock {\em Comb.}, 19(3):403--435, 1999.

\bibitem[Rou17]{R18}
Tim Roughgarden.
\newblock Complexity theory, game theory, and economics.
\newblock {\em Bellairs Research Institute of McGill University, Holetown,
  Barbados}, February 2017.
\newblock http://eccc.weizmann.ac.il/report/2018/001/.

\bibitem[Rub15]{R15}
Aviad Rubinstein.
\newblock Inapproximability of {Nash} equilibrium.
\newblock In {\em Proceedings of the Forty-Seventh Annual {ACM} on Symposium on
  Theory of Computing, {STOC} 2015, Portland, OR, USA, June 14-17, 2015}, pages
  409--418, 2015.

\bibitem[Rub16]{R16}
Aviad Rubinstein.
\newblock Settling the complexity of computing approximate two-player {Nash}
  equilibria.
\newblock In {\em {IEEE} 57th Annual Symposium on Foundations of Computer
  Science, {FOCS} 2016, 9-11 October 2016, Hyatt Regency, New Brunswick, New
  Jersey, {USA}}, pages 258--265, 2016.

\bibitem[RW16]{RW16}
Tim Roughgarden and Omri Weinstein.
\newblock On the communication complexity of approximate fixed points.
\newblock In {\em {IEEE} 57th Annual Symposium on Foundations of Computer
  Science, {FOCS} 2016, 9-11 October 2016, Hyatt Regency, New Brunswick, New
  Jersey, {USA}}, pages 229--238, 2016.

\bibitem[Sav18]{S18}
Neil Savage.
\newblock Always out of balance.
\newblock {\em Commun. {ACM}}, 61(4):12--14, 2018.

\bibitem[Spe28]{S28}
Emanuel Sperner.
\newblock Neuer beweis f\"{u}r die invarianz der dimensionszahl und des
  gebietes.
\newblock {\em Abh. Math. Sem. Hamburg}, VI:265--272, 1928.

\bibitem[Su99]{S99}
Francis Su.
\newblock Rental harmony: Sperner's lemma in fair division.
\newblock {\em The American Mathematical Monthly}, 106(10):930-- 942, 1999.

\bibitem[Whi34]{Whi34}
Hassler Whitney.
\newblock {Analytic extensions of differentiable functions defined in closed
  sets}.
\newblock {\em Transactions of the American Mathematical Society},
  36(1):63--89, 1934.

\bibitem[WW75]{WW75}
J.~H. Wells and L.~R. Williams.
\newblock {\em Embeddings and Extensions in Analysis}.
\newblock 1975.

\end{thebibliography}

\appendix

\section{From Nash Equilibrium to Brouwer Fixed Points}\label{sec:no}
In this section, we entertain the idea of trying to prove lower bounds similar to Theorem~\ref{thm:main} by combining the lower bound for computing a Nash equilibrium given in \cite{BR16}, with  the reduction from computing Nash equilibrium in games to finding fixed points in Brouwer functions given by Nash \cite{N51}.

For any $n$ player $m$ action game, the standard proof of Nash \cite{N51} produces a Brouwer function from $[0,1]^d$ to $[0,1]^d$ where $d=mn$. There are two critical issues with using this reduction.

First note that the input to a player in $n$ player $m$ action game is $m^n$ bits. The input to a player in the Brouwer problem is at least $2^d = 2^{mn}$ bits. In the case of $n=2$ (two-player games) this would yield an exponential blowup in the input size. Therefore, by using the lower bounds on 2 player Nash, we cannot hope to prove better lower bounds than logarithmic in the input size for the 2-player Brouwer problem. This should be compared to the polynomial lower bounds we were able to prove in this paper.

Second, one may consider starting from the $n$ agents binary action lower bound of \cite{BR16} in the two-player communication model (where each player knows the utility tensor of $n/2$ agents) and try to prove lower bounds for Brouwer. In this case, with some care, one can prove the lower bounds that we obtain for the various Brouwer problems in the $\ell_\infty$-norm. However, we elaborate below that the standard proof of Nash's theorem yields a Brouwer function with a high Lipschitz constant in the Euclidean norm (even when we consider the stronger lower bound of \cite{BR16} on computing  $(\varepsilon,\varepsilon)$-weak Nash equilibrium). 

Notice that following Nash's proof, the points of the compact convex space of the Brouwer function are in bijection with the space of mixed strategies and that the displacement on the $(i,\alpha)$ coordinate (where $(i,\alpha) \in [n] \times \{0,1\}$) of the constructed Brouwer function corresponds to the net gain of agent $i$ on unilaterally moving to action $\alpha$. This is $0$ if the point corresponds to a Nash equilibrium. We build utility tensors for each agent such that the corresponding Brouwer function has high Lipschitz constant. For any $x:=(x_1,\ldots ,x_n)\in \{0,1\}^n$, the utility of agent $i$ on playing $x_i$ and the rest playing $x_{-i}$ is 1 if $x_i=x_n=1$ and 0 otherwise. Now consider two pure strategies one in which all agents play $0$ (denoted by $x$) and the other in which all but agent $n$ plays $0$ and agent $n$ plays $1$ (denoted by $y$). The Euclidean distance between these two points (i.e., pure strategies) is $1/\sqrt{n}$. However, notice that if any of the agents unilaterally switches to the action 1, then the displacement of $y$ is 1 on the first $n-1$ coordinates. This means that the Lipschitz constant of the constructed Brouwer function is order $\sqrt{n}$. 

Summarizing, we show above that a black-box reduction from Nash problem to Brouwer problem cannot be utilized to obtain the results in this paper. It is entirely possible that the hard instances of \cite{BR16} do not have the above structure in its utility tensors, but that would likely require some non-trivial arguments.

\section{Communication Complexity Lower Bound for Finding Approximate Nash Equilibrium in Two-Player Games}\label{sec:Nash}
In this section, we provide a detailed proof outline of the following result of \cite{BR16}: The randomized communication complexity of finding an $\varepsilon$-Nash equilibrium in  two-player $N\times N$ games is $N^{\Omega(1)}$. 

Our starting point is the following variant of the End of a Line problem ($\EOL$). Let $H$ be a directed graph on vertex set $[N]$ and edge set $E$. We define $\EOL_H$ to be the problem where given as input $|E|$ bits describing a spanning subgraph $G([N],E')$  of $H$, the goal is to find a vertex $v\in [N]$ such that either:
\begin{description}
\item[Solution Type I:] $v=1$ and  in-deg($v$)$\neq$ 0 or out-deg($v$)$\neq 1$ in $G$; or,
\item[Solution Type II:] $v\neq 1$ and  in-deg($v$)$\neq$ 1 or out-deg($v$)$\neq 1$ in $G$.
\end{description}

Notice that we can always find some $v\in[N]$ such that one of the above two conditions hold. The complexity measure of the problem $\EOL_H$ that we are interested in studying is called the critical block sensitivity (\textsf{cbs}), a measure introduced in \cite{HN12} that
lower bounds randomized query complexity (among other things). We skip defining \textsf{cbs} formally here and point the reader to \cite{GR18}. Given the definition of \textsf{cbs}, it is a fairly simple exercise to show that $\EOL$ defined over the complete graph on $N$ vertices (where each vertex has both in-degree and out-degree to be $N$) has linear 
critical block sensitivity. 

\begin{proposition}
$\mathsf{cbs}(\EOL_{K_N})=\Omega(N)$. 
\end{proposition}

Next, as suggested in \cite{GR18}, 	we replace every vertex $v$  in $K_N$ by two complete binary trees $T^{\text{in}}_v$ and $T^{\text{out}}_v$ both with $N$ leaves where the edges in $T^{\text{in}}_v$ are all directed towards the root and the edges in $T^{\text{out}}_v$ are all directed away from the root. Then, for every $u,v\in[N]$ (not necessarily distinct) we have an edge from the $u^{\text{th}}$ leaf of $T^{\text{out}}_u$ to the $v^{\text{th}}$ leaf of $T^{\text{in}}_v$. Finally for every $v\in[N]$, we merge the roots of $T^{\text{in}}_v$ and $T^{\text{out}}_v$. The resulting graph (say $H$) has both in-degree and out-degree to be at most 2. Moreover, we can show that $\mathsf{cbs}(\EOL_{K_N})\le \mathsf{cbs}(\EOL_{H})$, but the vertex set of $H$ is $[4N^2-3N]$.

\begin{lemma}[Essentially \cite{GR18}]
There exists a graph $H([N],E)$ of in-degree and out-degree at most 2 for which we have
$\mathsf{cbs}(\EOL_H)=\Omega(\sqrt{N})$. 
\end{lemma}

At this point, we would like to move to the communication variant of $\EOL_H$ for a host graph $H([N],E)$. The problem is defined for a gadget function $g:\Sigma\times \Sigma\to\{0,1\}$ (for some alphabet set $\Sigma$). In $\EOL_H^g$, there are two players and each player is given $|E|\cdot |\Sigma|$ many bits as input, where we think of the input to each player as allocating $|\Sigma|$ many bits to each edge in $H$. Given $x\in \{0,1\}^{|E|\cdot |\Sigma|}$ to one player and $y\in \{0,1\}^{|E|\cdot |\Sigma|}$ to the other player as inputs, we define an underlying input graph $G$ as in $\EOL_H$ as follows: the $e^{\text{th}}$ edge of $H$ is present in $G$ if and only if $g(x\lvert_e,y\lvert_e)=1$. Their goal is to find $v\in[N]$ such that it is a solution of either type (I) or (II). By applying the simulation theorem of \cite{GP14a} on a constant sized gadget function $g$, we obtain a lower bound on the randomized communication complexity of $\EOL_H^g$.

\begin{theorem}[\cite{GP14a}]
There is a fixed alphabet set $\Sigma$ and a fixed gadget $g: \Sigma\times \Sigma\to\{0,1\}$ such that $\CC(\EOL_{H}^g)=\Omega(\mathsf{cbs}(\EOL_H))$.
\end{theorem}

Next, we interpret the inputs to both players in $\EOL_{H}^g$ as inputs to the Local Brouwer problem (see Definition~\ref{def:local}) by embedding the input graph $G$ (of $\EOL_H^g$ problem) into a (constant Lipschitz continuous) Brouwer function $f:[0,1]^n\to[0,1]^n$ in the Euclidean space ($n=O(\log N)$) using the embedding given in \cite{BR16}. Elaborating, the embedding of  \cite{BR16} ensures that the value of $f$ at any point in $[0,1]^n$ only depends on the information of the in-neighbors and out-neighbors of at most two vertices\footnote{Moreover, the embedding provides a function $\L$ which maps every point in $[0,1]^n$ to at most two vertices in $[N]$ (independent of the edge set of $G$).} in $G$. Since $H$ is of constant (in and out) degree, and $G$ is a subgraph of $H$, we have that the value of $f$ at any point in $[0,1]^n$ can be computed with some constant number of bits of communication. The embedding further guarantees that given any $\varepsilon$-approximate fixed point (for some small constant $\varepsilon>0$), we can recover a solution of $G$. 
This gives us a lower bound of $\Omega(\sqrt{N})=2^{\Omega(n)}$ on the randomized communication complexity of the Local Brouwer problem in $O(\log N)$ dimensions in the Euclidean metric.

\begin{theorem}[Restatement of Theorem~\ref{thm:BR17}]
There are fixed integers $r,\lambda$ and fixed constant $\varepsilon>0$  such that the randomized communication complexity of finding an $\varepsilon$-approximate fixed point of the Local Brouwer problem whose inputs are $r$-local and $\lambda$-Lipschitz in the $n$-dimensional Euclidean space is $2^{\Omega(n)}$.
\end{theorem}

Now we apply the reduction in the proof of Theorem~\ref{thm:loc2comp}, to obtain a lower bound of $2^{\Omega(n)}$ on the randomized communication complexity of the composition Brouwer problem (see Definition~\ref{def:AFP-RW}) in the $n$-dimensional Euclidean space. 

\begin{theorem}[Restatement of Theorem~\ref{thm:main}]\label{thm:brw}
There are fixed  constants $\lambda,\varepsilon>0$  such that the randomized communication complexity of finding an $\varepsilon$-approximate fixed point of the Composition Brouwer problem is $2^{\Omega(n)}$, where the input functions to each player is $\lambda$-Lipschitz and in $O(n)$-dimensional Euclidean space.
\end{theorem}

Finally, we  interpret the inputs to both of the players in the Composition Brouwer problem as inputs to the Nash equilibrium problem  by using the imitation gadget given in \cite{RW16}. Elaborating, given input $f_A:[0,1]^n\to[0,1]^{m=O(n)}$ to player $A$ (resp.\ $f_B:[0,1]^m\to[0,1]^{n}$ to player $B$), we first discretize the space $[0,1]^n$ (resp.\ $[0,1]^m$) using the discretization parameter $\alpha$, where $\alpha$ is smaller than $c\varepsilon/\lambda^2$, for some large constant $c$ and $\lambda:=\lambda_A\cdot \lambda_B$. Then the action space of player $A$ (resp.\ player $B$) is the discretized subset $[0,1]^n_\alpha$ (resp.\ $[0,1]^m_\alpha$). Player $A$  (resp.\ player $B$) then builds utility function $u_A$ (resp.\ $u_B$) over their action space $[0,1]^n_\alpha$ (resp.\ $[0,1]^m_\alpha$) as follows: $u_A(x,y)=-\|f_A(x)-y\|_2^2$ (resp.\ $u_B(x,y)=-\|x-f_B(y)\|_2^2$).
It can then be shown  that given any $O(\varepsilon^{4})$-approximate Nash equilibrium in the above two-player $N'\times N'$ game (where $N'=2^{O(n)}$), we can recover an $\varepsilon$-approximate fixed point of $f_B\circ f_A$. This gives us the lower bound of \cite{BR16}.

\begin{theorem}[Combining Theorem~\ref{thm:brw} with reduction from \cite{RW16}]
There is a constant $\varepsilon>0$  such that the randomized communication complexity of finding an $\varepsilon$-approximate Nash equilibrium in $N\times N$ two-player games is $N^{\Omega(1)}$.
\end{theorem}

\section{Total Regime}\label{sec:total}

In this section, we discuss for what range of parameters $\varepsilon,\lambda,$ and discretization parameter $\alpha$, are we in the total regime, i.e., we can guarantee an $\varepsilon$-approximate fixed point for a discretized Brouwer problem. This was explored for the $\ell_\infty$ norm by \cite{RW16}, and in this section we explore this question for every $\ell_p$ norm. To begin with, we need the following extension theorems.

For the max norm, Roughgarden and Weinstein \cite{RW16} provided a straightforward generalization of Whitney's extension theorem to higher dimensions as follows.

\begin{lemma}[\cite{Whi34,RW16}]\label{lem:extensionMax}
Let $\lambda\ge 0$. Let $n\in\mathbb{N}$ and $S\subseteq [0,1]^n$. For every $\lambda$-Lipschitz function $f:S\to [0,1]^n$ in the $\ell_\infty$ normed space, there exists a $\lambda$-Lipschitz function $\tilde f:[0,1]^n\to [0,1]^n$ in the $\ell_\infty$ normed space, such that for all $x\in S$ we have $f(x)=\tilde{f}(x)$.
\end{lemma}

Next, we adapt Kirszbraun's extension theorem \cite{K34} for the Euclidean norm.

\begin{lemma}\label{lem:extensionEuclid}
Let $\lambda\ge 0$. Let $n\in\mathbb{N}$ and $S\subseteq [0,1]^n$. For every $\lambda$-Lipschitz function $f:S\to [0,1]^n$ in the Euclidean space, there exists a $\lambda$-Lipschitz function $\tilde f:[0,1]^n\to [0,1]^n$ in the Euclidean space, such that for all $x\in S$ we have $f(x)=\tilde{f}(x)$.
\end{lemma}
\begin{proof}
Let $g:\mathbb{R}^n\to \mathbb{R}^n$ be a $\lambda$-Lipschitz function which is also an extension of $f$ guaranteed by  the Kirszbraun's extension theorem. Define $\tilde f:[0,1]^n\to[0,1]^n$ as follows: 
\[\forall x\in[0,1]^n\text{ and }i\in[n],\ \tilde{f}(x)_i=\begin{cases}
g(x)_i\text{ if }g(x)_i\in[0,1],\\
0\text{ if }g(x)_i<0,\\
1\text{ if }g(x)_i>1.
\end{cases}\]
It is easy to see that $\tilde{f}$ is also $\lambda$-Lipschitz.
\end{proof}

Now, we use the aforementioned extension theorems, to determine the total regime in the Euclidean and max norms.

\begin{theorem}\label{thm:total}
Let $p\in\{2,\infty\}$. Let $\varepsilon,\lambda\ge 0$, and $\alpha\in(0,1]$. Let $n\in\mathbb{N}$ and $f:[0,1]_\alpha^n\to [0,1]_\alpha^n$ be a $\lambda$-Lipschitz function in the $\ell_p$ normed space. If $(\lambda+1)\cdot \alpha\le 2\varepsilon$ then, $f$ has an $\varepsilon$-approximate fixed point.
\end{theorem}
\begin{proof}
Let $\tilde{f}:[0,1]^n\to [0,1]^n$ be the extension of $f$ guaranteed by Lemma~\ref{lem:extensionEuclid} if $p=2$ or by Lemma~\ref{lem:extensionMax} if $p=\infty$. By Brouwer's fixed point theorem, there exists $x\in[0,1]^n$ such that $\tilde f(x)=x$.
Let $y$ be an element in $[0,1]^n_\alpha$ that minimizes $\|x-y\|_p$ 
(if there are more than one, pick arbitrarily). 
Note that $\|x-y\|_p \leq \nicefrac{\alpha}{2}$. 
Since $\tilde{f}$ is $\lambda$-Lipschitz,
\begin{align*}
\|f(y) - y\|_p &\leq \|f(y) - \tilde f(x)\|_p + \|\tilde f(x) - y\|_p \\
&= \|\tilde f(y) - \tilde f(x)\|_p + \|x - y\|_p \\
&\leq \left(\lambda + 1\right) \cdot \|x - y\|_p \\
&\leq \left(\lambda + 1\right) \cdot \frac{\alpha}{2}\\
&\leq \varepsilon.\qedhere
\end{align*}
\end{proof}

We conclude this section  with a short discussion on determining the total regime for other $\ell_p$ norms. Fix a finite $p\in\mathbb{R}_{\ge 1}\setminus\{2\} $. If there was an extension theorem for the $\ell_p$ norm (similar to Lemmas~\ref{lem:extensionMax}~and~\ref{lem:extensionEuclid}) then the same proof of Theorem~\ref{thm:total} would give us conditions for the total regime in the $\ell_p$ norm. However it is well known \cite{WW75} that such an extension theorem cannot exist for every finite sized domain in the $\ell_p$ norm\footnote{In fact, strengthenings of this result are known; Naor \cite{N01} showed that even a non-isometric extension theorem cannot exist that holds for every finite sized domain in the $\ell_p$ norm when $p>1$, and Makarychev and Makarychev \cite{MM16} showed the same for $p=1$.}. Thus, if such an extension theorem existed for functions over the $[0,1]_\alpha^n$ domain then, they need to make use of the structure of the point-set $[0,1]_\alpha^n$. Thus, we leave open the following question.


\begin{open}
For any finite $p\in \mathbb{R}_{\ge 1}\setminus\{2\}$, is there any non-trivial setting of parameters $\varepsilon,\lambda,$ and $\alpha$ for which we can guarantee the existence of an $\varepsilon$-approximate fixed point in any of the Brouwer problems discussed in this paper?
\end{open}

\section{Connection to Monotone Karchmer-Wigderson Games}
\label{sec:KW}

	A natural question is whether the upper bound in Theorem~\ref{thm:14} can be improved, i.e., can we show  $\CC(\Sp^{d-1}_{d,n})=O(\log n)$ in the randomized communication complexity model.
	For example, for Karchmer-Wigderson (KW) games it was shown in \cite{JST11} (see also \cite{M17}) that any problem can be solved with $O(\log n)$ bits of communication.
	Our problem, however, would be equivalent to a \emph{monotone} KW game.
	In a monotone KW game, we are given a monotone Boolean function $g$ on $n$ variables, known to both players, and we have that for input $x\in\{0,1\}^n$ to player $A$ and  input $y\in\{0,1\}^n$ to player $B$,  $g(x)=1$ and $g(y)=0$ holds, respectively.
	Their goal is to find an $i\in[n]$ such that $x_i=1$ and $y_i=0$. 
	
	In our case, the variables $x_i$ can be the simplices of the triangulation, i.e., every $i$ corresponds to a simplex $i$.
	We define $g(x)=1$ if the set of simplices given by $\{i\mid x_i=1\}$, is a chain of simplices  from the facet $F_0$ to $F_d$, as described in the Surplus Sperner Lemma 
	(i.e., $g(x)=1$ if the vertices $f_0$ and $f_d$ are connected in the graph $G$ by a path whose vertices are all indexed by some $\{i\mid x_i=1\}$).
	
	Let $x_i=1$ if simplex $i$ has all colors from $1$ to $d-1$ (the colors known to player $A$) and its remaining two vertices are colored $0$ or $d$.
	This way $g(x)=1$ is exactly the statement of the Surplus Sperner Lemma.	
	
	Let $y_i=1$ if simplex $i$ has $d-1$ vertices colored  from $1$ to $d-1$ (the colors unknown to player $B$) and its remaining two vertices have the same color, i.e., both have color $0$ or $d$.
	This way $g(y)=0$ follows from the boundary conditions of the Sperner-coloring; if we had a path in $G$ from $f_0$ to $f_d$ formed by simplices from $\{i\mid y_i=1\}$, then since the first simplex would have need to have twice color $0$, the last simplex twice color $d$, and every chain of the simplex has only two vertices colored $0$ or $d$, in between there has to be a simplex with one of each color $0$ and $d$, which implies $y_i=0$.
	
	We have defined $x$ and $y$ such that a simplex is panchromatic if and only if $x_i=1$ and $y_i=0$.
	This means that the problem of finding a panchromatic simplex is exactly as hard as solving the monotone KW problem.
	In fact, it can be shown that our problem is equivalent to the randomized monotone circuit complexity of undirected $(s,t)$-connectivity, whose complexity is not known in literature.

\section{Connection to Hex Game}\label{sec:hex}
 
Finally, we would like to mention one more interesting connection and close this section with a small direction for future research.
Consider the following higher dimensional variant of the well-known Hex game, that is played on $\Delta^a \times \Delta^b$.

\begin{definition}[Hex($a,b$)]
	Two players claim the vertices of a triangulation of $\Delta^a \times \Delta^b$.
	Player A wins if her vertices span a $b$-manifold $M_A$ such that for every $q$ that is on the boundary of $\Delta^b$, there is a unique $p\in \Delta^a$ such that $(p,q)\in M_A$, and Player B wins if his vertices span an $a$-manifold $M_B$ such that for every $p$ that is on the boundary of $\Delta^a$, there is a unique $q\in \Delta^b$ such that $(p,q)\in M_B$.
\end{definition}

\begin{figure}
	\begin{center}
		\includegraphics[width=\linewidth]{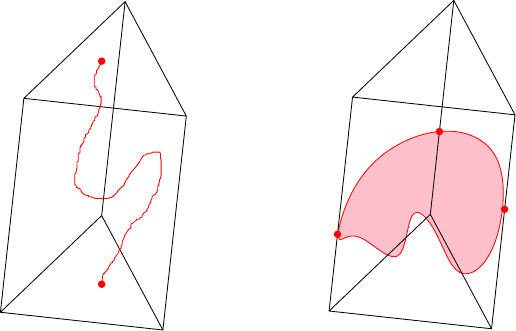}
		\caption{Hex($2,1$) game with A win to the left and B win to the right.}
		\label{fig:hex21}
	\end{center}
\end{figure}

Notice that $a=b=1$ is essentially just the usual Hex game \cite{N52,G79}, while $a=2, b=1$ is a game played on the triangulation of a triangular prism, where player $A$ is attempting to connect the top and the bottom triangular facets, while player $B$ is trying to make a surface whose boundary wraps around the quadrangular facets of the prism; see Figure \ref{fig:hex21}.
An existence theorem, similar to the usual argument for the Hex game and our Surplus Sperner Lemma\footnote{For a continuous analogue, see Exercise 3 on page 116 from \cite{M07} that says that if $f : S^k \to S^n$ and $g : S^\ell \to S^n$ are antipodal maps, then their images intersect.}, guarantees that exactly one player can win.
In fact, if we define $f_A$ and $f_B$ before the proof of Theorem \ref{thm:concSp2compBr} as $b$- and $a$-manifolds instead of functions, then they would be just the required $M_A$ and $M_B$. We remark that determining whether a position in a game of generalized Hex played on arbitrary graphs is a winning position is \textsf{PSPACE}-complete \cite{ET76}.

We can also define a monotone Boolean function, \textsf{HEX}, similarly as we did in the previous section.
The variables $x_i$ of \textsf{HEX} are indexed by the vertices of a triangulation of $\Delta^a \times \Delta^b$.
We define \textsf{HEX}$(x)=1$ if there is $b$-manifold $M_A$ from the vertices $\{i\mid x_i=1\}$ that is a win for player $A$ in the \textsf{HEX}($a,b$) game.
One can prove the monotone Karchmer-Wigderson complexity of \textsf{HEX} is the same as the complexity of the Concatenation Sperner problem, just as it was sketched in the previous section.
Therefore, Theorem \ref{thm:concSp2compBr} also implies that there is a monotone KW game whose randomized complexity is $n^{\Omega(1)}$, as opposed to non-monotone KW games, whose randomized complexity is always $O(\log n)$ \cite{JST11} (see also \cite{M17}).

This leads us to the following discussion. If it were possible to prove lower bounds to the above \textsf{HEX} game, directly in the communication model without relying on lifting/simulation theorems, then, we could reverse the direction of the above reductions, and obtain a lower bound for the concatenated Sperner problem (and consequently the problem of computing Nash equilibrium) without relying on lifting techniques. We leave this is an open direction of research.

\end{document}